\documentclass[11pt,amssymb,amsfont,a4paper]{article}
\usepackage{latexsym,amssymb,amsmath,amsthm,color}
\usepackage{geometry}
\usepackage{url}
\usepackage[numbers]{natbib}
\usepackage{graphicx}
\usepackage{color}
\usepackage[utf8]{inputenc}
\usepackage{authblk}
\usepackage{bbm}
\usepackage{tikz}
\usetikzlibrary{shapes,snakes}

\theoremstyle{definition}
\newtheorem{definition}{Definition}
\newtheorem{example}{Example}

\theoremstyle{plain}
\newtheorem{theorem}{Theorem}
\newtheorem{proposition}{Proposition}
\newtheorem{lemma}{Lemma}
\newtheorem{remark}{Remark}
\newtheorem{corollary}{Corollary}

\title{Bounding the number of common zeros of multivariate polynomials and their consecutive derivatives}
\author{Olav Geil \thanks{olav@math.aau.dk}}
\author{Umberto Mart{\'i}nez-Pe\~{n}as \thanks{umberto@math.aau.dk}}
\affil{Department of Mathematical Sciences, Aalborg University, Denmark}
\date{}

\begin{document}

\maketitle

\begin{abstract}
We upper bound the number of common zeros over a finite grid of multivariate polynomials and an arbitrary finite collection of their consecutive Hasse derivatives (in a coordinate-wise sense). To that end, we make use of the tool from Gr{\"o}bner basis theory known as footprint. Then we establish and prove extensions in this context of a family of well-known results in algebra and combinatorics. These include Alon's combinatorial Nullstellensatz \cite{alon}, existence and uniqueness of Hermite interpolating polynomials over a grid, estimations on the parameters of evaluation codes with consecutive derivatives \cite{multiplicitycodes}, and bounds on the number of zeros of a polynomial by DeMillo and Lipton \cite{demillo}, Schwartz \cite{schwartz}, Zippel \cite{zippel-first, zippel}, and Alon and F{\"u}redi \cite{alon-furedi}. As an alternative, we also extend the Schwartz-Zippel bound to weighted multiplicities and discuss its connection with our extension of the footprint bound.

\textbf{Keywords:} Footprint bound, Gr{\"o}bner basis, Hasse derivative, Hermite interpolation, multiplicity, Nullstellensatz, Schwartz-Zippel bound. 

\textbf{MSC:} 11T06, 12D10, 13P10. 
\end{abstract}

\section{Introduction}

Estimating the number of zeros of a polynomial over a field $ \mathbb{F} $ has been a central problem in algebra, where one of the main inconveniences is counting \textit{repeated zeros}, that is, \textit{multiplicities}. In the univariate case, this is easily solved by defining the multiplicity of a zero as the minimum positive integer $ r $ such that the first $ r $ \textit{consecutive derivatives} of the given polynomial vanish at that zero. In addition, Hasse derivatives \cite{hasse} are used instead of classical derivatives in order to give meaningful information over fields of positive characteristic. In this way, the number of zeros of a polynomial, counted with multiplicities, is upper bounded by its degree. Formally:
\begin{equation}
\sum_{a \in \mathbb{F}} m(F(x), a) \leq \deg(F(x)).
\label{eq in intro 1}
\end{equation}
If $ \mathcal{V}_{\geq r}(F(x)) $ denotes the set of zeros of $ F(x) $ of multiplicity at least $ r $, then a weaker, but still sharp, bound is the following:
\begin{equation}
\# \mathcal{V}_{\geq r}(F(x)) \cdot r \leq \deg(F(x)).
\label{eq in intro 2}
\end{equation}

In the multivariate case, the standard approach is to consider the first $ r $ consecutive Hasse derivatives as those whose multiindices have order less than $ r $, where the order of a multiindex $ (i_1, i_2, \ldots, i_m) $ is defined as $ \sum_{j=1}^m i_j $. We will use the terms \textit{standard multiplicities} to refer to this type of multiplicities. In this work, we consider arbitrary finite families $ \mathcal{J} $ of multiindices that are consecutive in a coordinate-wise sense: if $ (i_1, i_2, \ldots, i_m) $ belongs to $ \mathcal{J} $ and $ k_j \leq i_j $, for $ j = 1,2, \ldots, m $, then $ (k_1, k_2, \ldots, k_m) $ also belongs to $ \mathcal{J} $. Obviously, the (finite) family $ \mathcal{J} $ of multiindices of order less than a given positive integer $ r $ satisfies this property, hence is a particular case. 

Our main contribution is an upper bound on the number of common zeros over a grid of a family of polynomials and their (Hasse) derivatives corresponding to a finite set $ \mathcal{J} $ of consecutive multiindices. This upper bound makes use of the  technique from Gr{\"o}bner basis theory known as \textit{footprint} \cite{footprints, hoeholdt}, and can be seen as an extension of the classical \textit{footprint bound} \cite[Section 5.3]{clo1} in the sense of (\ref{eq in intro 2}). A first extension for standard multiplicities has been given as Lemma 2.4 in the expanded version of \cite{ruudbound}.

We will then show that this bound is sharp for ideals of polynomials, characterize those which satisfy equality, and give as applications extensions of known results in algebra and combinatorics: Alon's combinatorial Nullstellensatz \cite{alon, ball, clark, nulmultisets, shortproof}, existence and uniqueness of Hermite interpolating polynomials \cite{gasca, kopparty-multiplicity, lorentz}, estimations on the parameters of evaluation codes with consecutive derivatives \cite{weightedRM, kopparty-multiplicity, multiplicitycodes}, and the bounds by DeMillo and Lipton \cite{demillo}, Zippel \cite{zippel-first, zippel}, and Alon and F{\"u}redi \cite{alon-furedi}, and a particular case of the bound given by Schwartz in \cite[Lemma 1]{schwartz}. 

The bound in \cite[Lemma 1]{schwartz} can also be derived by those given by DeMillo and Lipton \cite{demillo}, and Zippel \cite[Theorem 1]{zippel-first}, \cite[Proposition 3]{zippel} (see Proposition \ref{proposition zippel implies schwartz} below), and is referred to as the \textit{Schwartz-Zippel bound} in many works in the literature \cite{extensions, weightedRM, kopparty-multiplicity, multiplicitycodes}. Interestingly, an extension of such bound for standard multiplicities in the sense of (\ref{eq in intro 1}) has been recently given in \cite[Lemma 8]{extensions}, but as Counterexample 7.4 in \cite{grid} shows, no straightforward extension of the footprint bound in the sense of (\ref{eq in intro 1}) seems possible (recall that we will give a footprint bound in the sense of (\ref{eq in intro 2})). To conclude this work, we give an extension of the Schwartz-Zippel bound in the sense of (\ref{eq in intro 1}) to derivatives with weighted order less than a given positive integer, which we will call \textit{weighted multiplicities}. This bound is inspired by \cite[Lemma 8]{extensions}, and we will discuss its connection with our extension of the footprint bound.

The results are organized as follows: We start with some preliminaries in Section 2. We then give the main bound in Section 3, together with some particular cases, an interpretation of the bound, and sharpness and equality conditions. In Section 4, we give a list of applications. Finally, in Section 5 we give an extension of the Schwartz-Zippel bound in the sense of (\ref{eq in intro 1}) to weighted multiplicities, and discuss the connections with the bound in Section 3.

\subsection*{Notation}

Throughout this paper, $ \mathbb{F} $ denotes an arbitrary field. We denote by $ \mathbb{F}[\mathbf{x}] = \mathbb{F}[x_1, x_2, \ldots, $ $ x_m] $ the ring of polynomials in the $ m $ variables $ x_1, x_2, \ldots, x_m $ with coefficients in $ \mathbb{F} $. A multiindex is a vector $ \mathbf{i} = (i_1, i_2, \ldots, i_m) \in \mathbb{N}^m $, where $ \mathbb{N} = \{ 0,1,2,3, \ldots \} $, and as usual we use the notation $ \mathbf{x}^\mathbf{i} = x_1^{i_1} x_2^{i_2} \cdots x_m^{i_m} $. We also denote $ \mathbb{N}_+ = \{ 1,2,3, \ldots \} $.

In this work, $ \preceq $ denotes the coordinate-wise partial ordering in $ \mathbb{N}^m $, that is, $ (i_1, i_2, \ldots, $ $ i_m) \preceq (j_1, j_2, \ldots, j_m) $ if $ i_k \leq j_k $, for all $ k = 1,2, \ldots, m $. We will use $ \preceq_m $ to denote a given monomial ordering in the set of monomials of $ \mathbb{F}[\mathbf{x}] $ (see \cite[Section 2.2]{clo1}), and we denote by $ {\rm LM}_{\preceq_m}(F(\mathbf{x})) $ the leading monomial of $ F(\mathbf{x}) \in \mathbb{F}[\mathbf{x}] $ with respect to $ \preceq_m $, or just $ {\rm LM}(F(\mathbf{x})) $ if there is no confusion about $ \preceq_m $. Finally, the notation $ \langle A \rangle $ means ideal generated by $ A $ in a ring, and $ \langle A \rangle_\mathbb{F} $ means vector space over $ \mathbb{F} $ generated by $ A $.

\section{Consecutive derivatives}

In this work, we consider Hasse derivatives, introduced first in \cite{hasse}. They coincide with usual derivatives except for multiplication with a non-zero constant factor when the corresponding multiindex contains no multiples of the characteristic of the field, and they have the advantage of not being identically zero otherwise. 

\begin{definition}[\textbf{Hasse derivative \cite{hasse}}] \label{def Hasse derivative}
Let $ F(\mathbf{x}) \in \mathbb{F}[\mathbf{x}] $ be a polynomial. Given another family of independent variables $ \mathbf{z} = (z_1, z_2, \ldots, z_m) $, the polynomial $ F(\mathbf{x} + \mathbf{z}) $ can be written uniquely as
$$ F(\mathbf{x} + \mathbf{z}) = \sum_{\mathbf{i} \in \mathbb{N}^m} F^{(\mathbf{i})}(\mathbf{x}) \mathbf{z}^\mathbf{i}, $$
for some polynomials $ F^{(\mathbf{i})}(\mathbf{x}) \in \mathbb{F}[\mathbf{x}] $, for $ \mathbf{i} \in \mathbb{N}^m $. For a given multiindex $ \mathbf{i} \in \mathbb{N}^m $, we define the $ \mathbf{i} $-th Hasse derivative of $ F(\mathbf{x}) $ as the polynomial $ F^{(\mathbf{i})}(\mathbf{x}) \in \mathbb{F}[\mathbf{x}] $.
\end{definition}

We next formalize the concept of zero of a polynomial of at least a given multiplicity as that of common zero of the given polynomial and a given finite family of its derivatives:

\begin{definition} \label{def general multiplicity}
Let $ F(\mathbf{x}) \in \mathbb{F}[\mathbf{x}] $ be a polynomial, let $ \mathbf{a} \in \mathbb{F}^m $ be an affine point, and let $ \mathcal{J} \subseteq \mathbb{N}^m $ be a finite set. We say that $ \mathbf{a} $ is a zero of $ F(\mathbf{x}) $ of multiplicity at least $ \mathcal{J} $ if $ F^{(\mathbf{i})}(\mathbf{a}) = 0 $, for all $ \mathbf{i} \in \mathcal{J} $. 
\end{definition}

The concept of \textit{consecutive derivatives}, in a coordinate-wise sense, can be formalized by the concept of \textit{decreasing sets} of multiindices (recall that $ \preceq $ denotes the coordinate-wise ordering in $ \mathbb{N}^m $):

\begin{definition} [\textbf{Decreasing sets}]
We say that the set $ \mathcal{J} \subseteq \mathbb{N}^m $ is decreasing if whenever $ \mathbf{i} \in \mathcal{J} $ and $ \mathbf{j} \in \mathbb{N}^m $ are such that $ \mathbf{j} \preceq \mathbf{i} $, it holds that $ \mathbf{j} \in \mathcal{J} $.
\end{definition}

Observe that the finite set $ \mathcal{J} = \{ (i_1, i_2, \ldots, i_m) \in \mathbb{N}^m : \sum_{j=1}^m i_j < r \} $, for a positive integer $ r $, is decreasing. Moreover, if $ m = 1 $, then these are all possible decreasing finite sets. The concept of weighted orders and weighted multiplicities shows that this is not the case when $ m > 1 $:

\begin{definition} [\textbf{Weighted multiplicities}] \label{def weighted multiplicity}
Fix a vector of positive weights $ \mathbf{w} = (w_1, $ $ w_2, $ $ \ldots, w_m) \in \mathbb{N}_+^m $. Given a multiindex $ \mathbf{i} = (i_1, i_2, \ldots, i_m) \in \mathbb{N}^m $, we define its weighted order as
\begin{equation}
 \mid \mathbf{i} \mid_\mathbf{w} = i_1 w_1 + i_2 w_2 + \cdots + i_m w_m.
\label{eq def weighted norm}
\end{equation}
Let $ F(\mathbf{x}) \in \mathbb{F}[\mathbf{x}] $ be a polynomial and let $ \mathbf{a} \in \mathbb{F}^m $ be an affine point. We say that $ \mathbf{a} $ is a zero of $ F(\mathbf{x}) $ of weighted multiplicity $ r \in \mathbb{N} $, and we write
$$ m_\mathbf{w}(F(\mathbf{x}), \mathbf{a}) = r, $$
if $ F^{(\mathbf{i})}(\mathbf{a}) = 0 $, for all $ \mathbf{i} \in \mathbb{N}^m $ with $ \mid \mathbf{i} \mid_{\mathbf{w}} < r $, and $ F^{(\mathbf{j})}(\mathbf{a}) \neq 0 $, for some $ \mathbf{j} \in \mathbb{N}^m $ with $ \mid \mathbf{j} \mid_{\mathbf{w}} = r $.
\end{definition}

We also introduce the definition of weighted degree, which will be convenient for different results in the following sections:

\begin{definition} [\textbf{Weighted degrees}] \label{def weighted degree}
Let $ F(\mathbf{x}) \in \mathbb{F}[\mathbf{x}] $ be a polynomial and let $ \mathbf{w} \in \mathbb{N}_+^m $ be a vector of positive weights. We define the weighted degree of $ F(\mathbf{x}) $ as 
$$ \deg_\mathbf{w}(F(\mathbf{x})) = \max \{ \mid \mathbf{i} \mid_\mathbf{w} : F_\mathbf{i} \neq 0 \}, $$
where $ F(\mathbf{x}) = \sum_{\mathbf{i} \in \mathbb{N}^m} F_\mathbf{i} \mathbf{x}^\mathbf{i} $ and $ F_\mathbf{i} \in \mathbb{F} $, for all $ \mathbf{i} \in \mathbb{N}^m $.
\end{definition}

Other interesting sets of consecutive derivatives that we will consider throughout the paper are those given by bounding each index separately, that is, sets of the form $ \mathcal{J} = \left\lbrace (i_1, i_2, \ldots, i_m) \in \mathbb{N}^m : i_j < r_j, j = 1,2, \ldots, m \right\rbrace $, for a given $ (r_1, r_2, \ldots, r_m) \in \mathbb{N}_+^m $, where $ \preceq $ denotes the coordinate-wise partial ordering.

\section{The footprint bound for consecutive derivatives} \label{sec footprint bounds}

In this section, we will give an extension of the footprint bound \cite[Section 5.3]{clo1} to upper bound the number of common zeros over a finite grid of a family of polynomials and a given set of their consecutive derivatives, as in Definition \ref{def general multiplicity}. We give some particular cases and an interpretation of the bound. We conclude by studying its sharpness.

Throughout the section, fix a decreasing finite set $ \mathcal{J} \subseteq \mathbb{N}^m $, an ideal $ I \subseteq \mathbb{F}[\mathbf{x}] $ and finite subsets $ S_1, S_2, \ldots, S_m \subseteq \mathbb{F} $. Write $ S = S_1 \times S_2 \times \cdots \times S_m $, and denote by $ G_j(x_j) \in \mathbb{F}[x_j] $ the defining polynomial of $ S_j $, that is, $ G_j(x_j) = \prod_{s \in S_j}(x_j-s) $, for $ j = 1,2, \ldots, m $. The three objects involved in our bound are the following:

\begin{definition} \label{def main objects in bound}
We define the ideal 
$$ I_{\mathcal{J}} = I + \left\langle \left\lbrace \prod_{j=1}^m G_j(x_j)^{r_j} : (r_1, r_2, \ldots, r_m) \notin \mathcal{J} \right\rbrace \right\rangle $$
and the set of zeros of multiplicity at least $ \mathcal{J} $ of the ideal $ I $ in the grid $ S = S_1 \times S_2 \times \cdots \times S_m $ as
$$ \mathcal{V}_{\mathcal{J}}(I) = \left\lbrace \mathbf{a} \in S : F^{(\mathbf{i})}(\mathbf{a}) = 0, \forall F(\mathbf{x}) \in I, \forall \mathbf{i} \in \mathcal{J} \right\rbrace . $$
Finally, given a monomial ordering $ \preceq_m $, we define the footprint of an ideal $ J \subseteq \mathbb{F}[\mathbf{x}] $ as
$$ \Delta_{\preceq_m} (J) = \left\lbrace \mathbf{x}^\mathbf{i} : \mathbf{x}^\mathbf{i} \notin \left\langle {\rm LM}(J) \right\rangle \right\rbrace , $$
where $ {\rm LM}(J) = \{ {\rm LM}(F(\mathbf{x})) : F(\mathbf{x}) \in J \} $ with respect to the monomial ordering $ \preceq_m $. We write $ \Delta(J) $ if there is no confusion about the monomial ordering.
\end{definition}

\subsection{The general bound}

\begin{theorem} \label{th footprint bound}
For any monomial ordering, it holds that
\begin{equation} \label{general bound}
\# \mathcal{V}_{\mathcal{J}}(I) \cdot \# \mathcal{J} \leq \# \Delta \left( I_\mathcal{J} \right).
\end{equation}
\end{theorem}

The rest of the subsection is devoted to the proof of this result. The first auxiliary tool is the Leibniz formula, which follows by a straightforward computation (see also \cite[pages 144--155]{torresbook}):

\begin{lemma} [\textbf{Leibniz formula}] \label{lemma leibniz formula}
Let $ F_1(\mathbf{x}), F_2(\mathbf{x}), \ldots, F_s(\mathbf{x}) \in \mathbb{F}[\mathbf{x}] $ and let $ \mathbf{i} \in \mathbb{N}^m $. It holds that
$$ \left( \prod_{j=1}^s F_j(\mathbf{x}) \right)^{(\mathbf{i})} = \sum_{\mathbf{i}_1 + \mathbf{i}_2 + \cdots + \mathbf{i}_s = \mathbf{i}} \left( \prod_{j=1}^s F_j^{(\mathbf{i}_j)}(\mathbf{x}) \right). $$
\end{lemma}

The second auxiliary tool is the existence of Hermite interpolating polynomials with Hasse derivatives. For our purposes, a \textit{separated-variables} extension of univariate Hermite interpolation over grids is enough. This extension is straightforward and seems to be known in the literature (see \cite[Section 3.1]{lorentz}), but we give a short proof in the Appendix for convenience of the reader.

\begin{definition} \label{def evaluation map}
We define the evaluation map on a finite set $ T \subseteq \mathbb{F}^m $ with derivatives corresponding to multiindices in $ \mathcal{J} $ as
\begin{equation} \label{evaluation map definition}
\begin{split}
{\rm Ev} : & \mathbb{F}[\mathbf{x}] \longrightarrow \mathbb{F}^{\# T \cdot \# \mathcal{J}} \\
 & F(\mathbf{x}) \mapsto \left( \left( F^{(\mathbf{i})}(\mathbf{a}) \right) _{\mathbf{i} \in \mathcal{J}} \right) _{\mathbf{a} \in T} . 
\end{split}
\end{equation}
\end{definition}

\begin{lemma}[\textbf{Hermite interpolation}] \label{lemma hermite}
The evaluation map $ {\rm Ev} : \mathbb{F}[\mathbf{x}] \longrightarrow \mathbb{F}^{\# T \cdot \# \mathcal{J}} $ defined in (\ref{evaluation map definition}) is surjective, for all finite sets $ T \subseteq \mathbb{F}^m $ and $ \mathcal{J} \subseteq \mathbb{N}^m $.
\end{lemma}
\begin{proof}
See the Appendix. 
\end{proof}

With these tools, we may now prove Theorem \ref{th footprint bound}:

\begin{proof}[Proof of Theorem \ref{th footprint bound}]
Fix multiindices $ \mathbf{r} = (r_1, r_2, \ldots, r_m) \notin \mathcal{J} $ and $ \mathbf{i} = (i_1, i_2, \ldots, i_m) \in \mathcal{J} $, and define $ G(\mathbf{x}) = \prod_{j=1}^m G_j(x_j)^{r_j} $. By Lemma \ref{lemma leibniz formula}, it holds that
\begin{equation}
 G^{(\mathbf{i})}(\mathbf{x}) = \prod_{j=1}^m \left( G_j(x_j)^{r_j} \right)^{(i_j)}.
\label{eq in proof of footprint 1}
\end{equation}
Furthermore, if $ r > i $ and $ F(x) \in \mathbb{F}[x] $, then there exists $ H(x) \in \mathbb{F}[x] $ such that
\begin{equation}
\left( F(x)^r \right)^{(i)} = \sum_{i_1 + i_2 + \cdots + i_r = i} \left( \prod_{j=1}^r F^{(i_j)}(x) \right) = H(x) F(x)^{r-i},
\label{eq in proof of footprint 2}
\end{equation}
again by Lemma \ref{lemma leibniz formula}, since at least $ r-i > 0 $ indices $ i_j $ must be equal to $ 0 $, for each $ (i_1, i_2, \ldots, i_m) \in \mathbb{N}^m $ such that $ \sum_{j=1}^m i_j = i $. Finally, since $ \mathcal{J} $ is decreasing, it holds that $ \mathbf{r} - \mathbf{i} $ has at least one positive coordinate. Hence, combining (\ref{eq in proof of footprint 1}) and (\ref{eq in proof of footprint 2}), we see that $ G^{(\mathbf{i})}(\mathbf{a}) = 0 $, for all $ \mathbf{a} \in \mathcal{V}_{\mathcal{J}}(I) \subseteq S $. This implies that
$$ {\rm Ev}(F(\mathbf{x})) = \mathbf{0}, \quad \forall F(\mathbf{x}) \in I_\mathcal{J}, $$
by the definition of the ideal $ I_\mathcal{J} $ and the set $ \mathcal{V}_{\mathcal{J}}(I) $, and where we consider $ T = \mathcal{V}_{\mathcal{J}}(I) $ in the definition of $ {\rm Ev} $ (Definition \ref{def evaluation map}).

Therefore, the evaluation map $ {\rm Ev} $ can be extended to the quotient ring
$$ {\rm Ev} : \mathbb{F}[\mathbf{x}] / I_\mathcal{J} \longrightarrow \mathbb{F}^{\# \mathcal{V}_{\mathcal{J}}(I) \cdot \# \mathcal{J}}, $$
which is again surjective, since the original evaluation map is surjective by Lemma \ref{lemma hermite}. Since the domain and codomain of this map are $ \mathbb{F} $-linear vector spaces and the map itself is also $ \mathbb{F} $-linear, we conclude that
$$ \# \mathcal{V}_{\mathcal{J}}(I) \cdot \# \mathcal{J} = \dim_\mathbb{F} \left( \mathbb{F}^{\# \mathcal{V}_{\mathcal{J}}(I) \cdot \# \mathcal{J}} \right) \leq \dim_\mathbb{F} \left( \mathbb{F}[\mathbf{x}] / I_\mathcal{J} \right). $$
Finally, Proposition 4 in \cite[Section 5.3]{clo1} says that the monomials in $ \Delta(J) $ constitute a basis of $ \mathbb{F}[\mathbf{x}] / J $, for an ideal $ J \subseteq \mathbb{F}[\mathbf{x}] $. This fact implies that
$$ \dim_\mathbb{F} \left( \mathbb{F}[\mathbf{x}] / I_\mathcal{J} \right) = \# \Delta \left( I_\mathcal{J} \right), $$
and the result follows.
\end{proof}

\subsection{Some particular cases}

In this subsection, we derive some particular cases of Theorem \ref{th footprint bound}. We start with the classical form of the footprint bound (see Proposition 8 in \cite[Section 5.3]{clo1}, and \cite{footprints, hoeholdt}):

\begin{corollary} [\textbf{\cite{clo1, footprints, hoeholdt}}] \label{usual footprint bound}
Setting $ \mathcal{J} = \{ \mathbf{0} \} $, we obtain that
$$ \#  \mathcal{V}(I) \leq \# \Delta \left( I + \left\langle G_1(x_1), G_2(x_2), \ldots, G_m(x_m) \right\rangle \right), $$
where $ \mathcal{V}(I) $ denotes the set of zeros of the ideal $ I $ in $ S $.
\end{corollary}

The case of zeros of standard multiplicity at least a given positive integer was first obtained as Lemma 2.4 in the extended version of \cite{ruudbound}, and reads as follows:

\begin{corollary}[\textbf{\cite{ruudbound}}] \label{footprint bound with multi}
Given an integer $ r \in \mathbb{N}_+ $, and setting $ \mathcal{J} = \{ (i_1, i_2, \ldots, i_m) \in \mathbb{N}^m : \sum_{j=1}^m i_j < r \} $, we obtain that
$$ \#  \mathcal{V}_{\geq r}(I) \cdot \binom{m + r - 1}{m} \leq \# \Delta \left( I + \left\langle \left\lbrace \prod_{j=1}^m G_j(x_j)^{r_j} : \sum_{j=1}^m r_j = r \right\rbrace \right\rangle \right), $$
where $ \mathcal{V}_{\geq r}(I) $ denotes the set of zeros of multiplicity at least $ r $ of the ideal $ I $ in $ S $.
\end{corollary}

Another particular case is obtained when upper bounding each coordinate of the multiindices separately:

\begin{corollary} 
Given a multiindex $ (r_1, r_2, \ldots, r_m) \in \mathbb{N}_+^m $, and setting $ \mathcal{J} = \{ (i_1, i_2, \ldots, i_m) $ $ \in \mathbb{N}^m : i_j < r_j, j = 1,2, \ldots, m \} $, we obtain that
$$ \# \mathcal{V}_{\mathcal{J}}(I) \cdot \prod_{j=1}^m r_j  \leq \# \Delta \left( I + \left\langle G_1(x_1)^{r_1}, G_2(x_2)^{r_2}, \ldots, G_m(x_m)^{r_m} \right\rangle \right). $$
\end{corollary}

Finally, we obtain a footprint bound for weighted multiplicities:

\begin{corollary} \label{corollary footprint for weighted multi}
Given an integer $ r \in \mathbb{N}_+ $, a vector of positive weights $ \mathbf{w} = (w_1, w_2, \ldots, $ $ w_m) $ $ \in \mathbb{N}_+ $, and setting $ \mathcal{J} = \{ \mathbf{i} \in \mathbb{N}^m : \mid \mathbf{i} \mid_{\mathbf{w}} < r \} $, we obtain that
$$ \#  \mathcal{V}_{\geq r, \mathbf{w}}(I) \cdot {\rm B}(\mathbf{w}; r) \leq \# \Delta \left( I + \left\langle \left\lbrace \prod_{j=1}^m G_j(x_j)^{r_j} : \sum_{j=1}^m r_jw_j \geq r \right\rbrace \right\rangle \right), $$
where $ \mathcal{V}_{\geq r, \mathbf{w}}(I) $ denotes the set of zeros of weighted multiplicity at least $ r $ of the ideal $ I $ in $ S $, and where $ {\rm B}(\mathbf{w}; r) = \# \left\lbrace \mathbf{i} \in \mathbb{N}^m : \mid \mathbf{i} \mid_{\mathbf{w}} < r \right\rbrace $.
\end{corollary}

To conclude, we give a more explicit form of the bound in the previous corollary by estimating the number $ B(\mathbf{w}; r) $:

\begin{corollary}
Given an integer $ r \in \mathbb{N}_+ $ and a vector of positive weights $ \mathbf{w} = (w_1, w_2, \ldots, $ $ w_m) $ $ \in \mathbb{N}_+ $, it holds that
\begin{equation}
\binom{m + r - 1}{m} \leq w_1 w_2 \cdots w_m B(\mathbf{w}; r).
\label{eq estimate on weighted binomial coeffs}
\end{equation}
In particular, we deduce from the previous corollary that
$$ \#  \mathcal{V}_{\geq r, \mathbf{w}}(I) \cdot \binom{m + r - 1}{m} \leq w_1 w_2 \cdots w_m \cdot \# \Delta \left( I + \left\langle \left\lbrace \prod_{j=1}^m G_j(x_j)^{r_j} : \sum_{j=1}^m r_jw_j \geq r \right\rbrace \right\rangle \right). $$
\end{corollary}
\begin{proof}
Define the map $ T_\mathbf{j} : \mathbb{N}^m \longrightarrow \mathbb{N}^m $ by 
$$ T_\mathbf{j}(\mathbf{i}) = (i_1w_1 + j_1, i_2w_2 + j_2, \ldots, i_mw_m + j_m), $$
for all $ \mathbf{i} = (i_1, i_2, \ldots, i_m), \mathbf{j} = (j_1, j_2, \ldots, j_m) \in \mathbb{N}^m $. Now define $ \mathcal{J}(\mathbf{w}; r) = \{ \mathbf{i} \in \mathbb{N}^m : \\ \mid \mathbf{i} \mid_{\mathbf{w}} < r \} $. By the Euclidean division, we see that
$$ \mathcal{J}((1,1, \ldots, 1); r) \subseteq \bigcup_{\mathbf{j} \in \prod_{k=1}^m [0,w_k)} T_\mathbf{j} \left( \mathcal{J} \left( \mathbf{w}; r \right) \right) . $$
By counting elements on both sides of the inclusion, the result follows.
\end{proof}

\subsection{Interpretation of the bound and illustration of the set $ \Delta(I_\mathcal{J}) $} \label{subsec interpretation bound}

In this subsection, we give a graphical description of the footprint $ \Delta(I_\mathcal{J}) $ which will allow us to provide an interpretation of the bound (\ref{general bound}).

First, we observe that by adding the polynomials $ \prod_{i=1}^m G_i(x_i)^{r_i} $, for $ (r_1, r_2, \ldots, r_m) \notin \mathcal{J} $, we are bounding the set of points $ \Delta(I_\mathcal{J}) $ by a certain subset $ \mathcal{J}_S \subseteq \mathbb{N}^m $, which we now define:

\begin{definition}
We define the set
$$ \mathcal{J}_S =  \left\lbrace \mathbf{i} \in \mathbb{N}^m : \mathbf{i} \nsucceq (r_1 \# S_1, r_2 \# S_2, \ldots, r_m \# S_m), \forall (r_1, r_2, \ldots, r_m) \notin \mathcal{J} \right\rbrace . $$
\end{definition}

For clarity, we now give a description of this set by a positive defining condition that follows from the properties of the Euclidean division and the fact that $ \mathcal{J} $ is decreasing.

\begin{lemma} \label{lemma auxiliary applications footprint 0} 
It holds that
\begin{equation*}
\begin{split}
\mathcal{J}_S = \{ & \left( p_1 \# S_1 + t_1, p_2 \# S_2 + t_2, \ldots, p_m \# S_m + t_m \right) \in \mathbb{N}^m : \\
 & (p_1, p_2, \ldots, p_m) \in \mathcal{J}, 0 \leq t_j < \# S_j, \forall j = 1,2, \ldots, m \}.
\end{split}
\end{equation*}
\end{lemma}

We may then state the fact that the footprint is bounded by this set as follows:

\begin{lemma}
It holds that
$$ \Delta(I_\mathcal{J}) \subseteq \{ \mathbf{x}^\mathbf{i} : \mathbf{i} \in \mathcal{J}_S \}. $$
\end{lemma}

Moreover, the set $ \mathcal{J}_S $ can be easily seen as the union of $ \# \mathcal{J} $ $ m $-dimensional rectangles in $ \mathbb{N}^m $ whose sides have lengths $ \# S_1 $, $ \# S_2 $, $ \ldots $, $ \# S_m $, respectively. In particular, we obtain the following:

\begin{lemma}\label{lemma auxiliary applications footprint 1} 
It holds that
\begin{equation}
\# \mathcal{J}_S = \# S \cdot \# \mathcal{J}.
\label{eq computation elements not in}
\end{equation}
\end{lemma}

The footprint bound (\ref{general bound}) can then be interpreted as follows: Consider the set $ \mathcal{J}_S \subseteq \mathbb{N}^m $. For each $ \mathbf{x}^\mathbf{i} \in {\rm LM}(I_\mathcal{J}) $, remove from $ \mathcal{J}_S $ all points $ \mathbf{j} $ such that $ \mathbf{i} \preceq \mathbf{j} $. The remaining points correspond to the multiindices in $ \Delta(I_\mathcal{J}) $, and thus there are $ \# \Delta(I_\mathcal{J}) $ of them.

In particular, if $ F_1(\mathbf{x}), F_2(\mathbf{x}), \ldots, F_t(\mathbf{x}) \in I $, then we may only remove the points corresponding to $ {\rm LM}(F_i(\mathbf{x})) $, for $ i = 1,2, \ldots, t $, and we obtain an upper bound on $ \# \Delta(I_\mathcal{J}) $. 

\begin{example}
Let us assume now that $ m = 2 $, $ \# S_1 = \# S_2 = 2 $, and $ \mathcal{J} = \{ (0,1), $ $ (1,1),$ $ (2,1),$ $ (0,0), $ $ (1,0), $ $ (2,0), $ $ (3,0), $ $ (4,0), $ $ (5,0) \} $.

In Figure \ref{fig interpretation}, top image, we represent by black dots the monomials whose multiindices belong to $ \mathcal{J}_S $, among which medium-sized dots correspond to multiindices that belong to $ \mathcal{J} $ when each coordinate is multiplied by $ 2 $. Blank dots correspond to multiindices that do not belong to $ \mathcal{J}_S $, and the largest ones correspond to minimal multiindices that do not belong to $ \mathcal{J}_S $.

In Figure \ref{fig interpretation}, bottom image, we represent in the same way the set $ \Delta(I_\mathcal{J}) $, whenever $ \langle {\rm LM}(I_\mathcal{J}) \rangle $ is generated by $ x_1^2x_2^3 $, $ x_1^8x_2 $, and the leading monomials of $ G_1(x_1)^{r_1} G_2(x_2)^{r_2} $, for minimal $ (r_1,r_2) \notin \mathcal{J} $, which in this case are $ x_2^4 $, $ x_1^6x_2^2 $ and $ x_1^{12} $.

In conclusion, the bound (\ref{general bound}) says that the number of zeros in $ S $ of $ I $ of multiplicity at least $ \mathcal{J} $ is at most $ 3 $.
\end{example}

\begin{figure}[h]

\begin{center}

\begin{tabular}{c@{\extracolsep{1cm}}c}
	\begin{tikzpicture}[line width=1pt, scale=1]
		\tikzstyle{every node}=[inner sep=0pt, minimum width=4.5pt]
		
		\draw (0,0) node (O) [draw, circle, fill=black] {};
		\draw (0,5) node (topy) {};
		\draw (14,0) node (topx) {};
		\draw[->] (O) -- (topy);
		\draw[->] (O) -- (topx);
		
		\foreach \p in {(0,4),(6,2),(12,0)}{\draw \p node [draw, scale=2, circle, fill=white] {};}		
		
		\foreach \p in {(0,2),(2,2),(2,0),(4,2),(4,0),(6,0),(8,0),(10,0)}{\draw \p node [draw, circle, fill=black] {};}
		\foreach \p in {(2,4),(4,4),(6,4),(8,2),(10,2),(12,2)}{\draw \p node [draw, circle, fill=white] {};}
		
		\foreach \p in {(0,3),(0,1),(1,3),(1,2),(1,1),(1,0),(2,3),(2,1),(3,3),(3,2),(3,1),(3,0),(4,3),(4,1),(5,3),(5,2),(5,1),(5,0),(6,1),(7,1),(7,0),(8,1),(9,1),(9,0),(10,1),(11,1),(11,0)}{\draw \p node [draw, scale=0.5, circle, fill=black] {};}
		\foreach \p in {(1,4),(3,4),(5,4),(6,3),(7,2),(9,2),(11,2),(12,1)}{\draw \p node [draw, scale=0.7, circle, fill=white] {};}
		
		\draw (-1,3.5) -- (5.5,3.5) -- (5.5,1.5) -- (11.5,1.5) -- (11.5,-1);
		
		\draw (-0.7,-0.7) node {$ 1 $};
		\draw (-0.7,1) node {$ x_2 $};
		\draw (-0.7,2) node {$ x_2^2 $};
		\draw (-0.7,3) node {$ x_2^3 $};
		\draw (-0.7,4) node {$ x_2^4 $};
		
		\draw (1,-0.7) node {$ x_1 $};
		\draw (2,-0.7) node {$ x_1^2 $};
		\draw (3,-0.7) node {$ x_1^3 $};
		\draw (4,-0.7) node {$ x_1^4 $};
		\draw (5,-0.7) node {$ x_1^5 $};
		\draw (6,-0.7) node {$ x_1^6 $};
		\draw (7,-0.7) node {$ x_1^7 $};
		\draw (8,-0.7) node {$ x_1^8 $};
		\draw (9,-0.7) node {$ x_1^9 $};
		\draw (10,-0.7) node {$ x_1^{10} $};
		\draw (11,-0.7) node {$ x_1^{11} $};
		\draw (12,-0.7) node {$ x_1^{12} $};
		
		\draw (10,4) node {The set $ \mathcal{J}_S $};
		
	\end{tikzpicture}
\end{tabular}

\begin{tabular}{c@{\extracolsep{1cm}}c}
	\begin{tikzpicture}[line width=1pt, scale=1]
		\tikzstyle{every node}=[inner sep=0pt, minimum width=4.5pt]
		
		\draw (0,0) node (O) [draw, circle, fill=black] {};
		\draw (0,5) node (topy) {};
		\draw (14,0) node (topx) {};
		\draw[->] (O) -- (topy);
		\draw[->] (O) -- (topx);
		
		\foreach \p in {(0,4),(6,2),(12,0),(2,3),(8,1)}{\draw \p node [draw, scale=2, circle, fill=white] {};}		
		
		\foreach \p in {(0,2),(2,2),(2,0),(4,2),(4,0),(6,0),(8,0),(10,0)}{\draw \p node [draw, circle, fill=black] {};}
		\foreach \p in {(2,4),(4,4),(6,4),(8,2),(10,2),(12,2)}{\draw \p node [draw, circle, fill=white] {};}
		
		\foreach \p in {(0,3),(0,1),(1,3),(1,2),(1,1),(1,0),(2,1),(3,2),(3,1),(3,0),(4,1),(5,2),(5,1),(5,0),(6,1),(7,1),(7,0),(9,0),(11,0)}{\draw \p node [draw, scale=0.5, circle, fill=black] {};}
		\foreach \p in {(1,4),(3,3),(3,4),(4,3),(5,3),(5,4),(6,3),(7,2),(9,1),(9,2),(10,1),(11,1),(11,2),(12,1)}{\draw \p node [draw, scale=0.7, circle, fill=white] {};}
		
		\draw (-1,3.5) -- (1.5,3.5) -- (1.5,2.5) -- (5.5,2.5) -- (5.5,1.5) -- (7.5,1.5) -- (7.5,0.5) -- (11.5,0.5) -- (11.5,-1);
		\draw [dashed] (1.5,3.5) -- (5.5,3.5) -- (5.5,2.5); 
		\draw [dashed] (7.5,1.5) -- (11.5,1.5) -- (11.5,0.5); 		
		
		\draw (-0.7,-0.7) node {$ 1 $};
		\draw (-0.7,1) node {$ x_2 $};
		\draw (-0.7,2) node {$ x_2^2 $};
		\draw (-0.7,3) node {$ x_2^3 $};
		\draw (-0.7,4) node {$ x_2^4 $};
		
		\draw (1,-0.7) node {$ x_1 $};
		\draw (2,-0.7) node {$ x_1^2 $};
		\draw (3,-0.7) node {$ x_1^3 $};
		\draw (4,-0.7) node {$ x_1^4 $};
		\draw (5,-0.7) node {$ x_1^5 $};
		\draw (6,-0.7) node {$ x_1^6 $};
		\draw (7,-0.7) node {$ x_1^7 $};
		\draw (8,-0.7) node {$ x_1^8 $};
		\draw (9,-0.7) node {$ x_1^9 $};
		\draw (10,-0.7) node {$ x_1^{10} $};
		\draw (11,-0.7) node {$ x_1^{11} $};
		\draw (12,-0.7) node {$ x_1^{12} $};
		
		\draw (10,4) node {The set $ \Delta(I_\mathcal{J}) $};
		
	\end{tikzpicture}
\end{tabular}

\end{center}
\caption{Illustration of the sets $ \mathcal{J}_S $ and $ \Delta(I_\mathcal{J}) $ in $ \mathbb{N}^m $.}
\label{fig interpretation}
\end{figure}
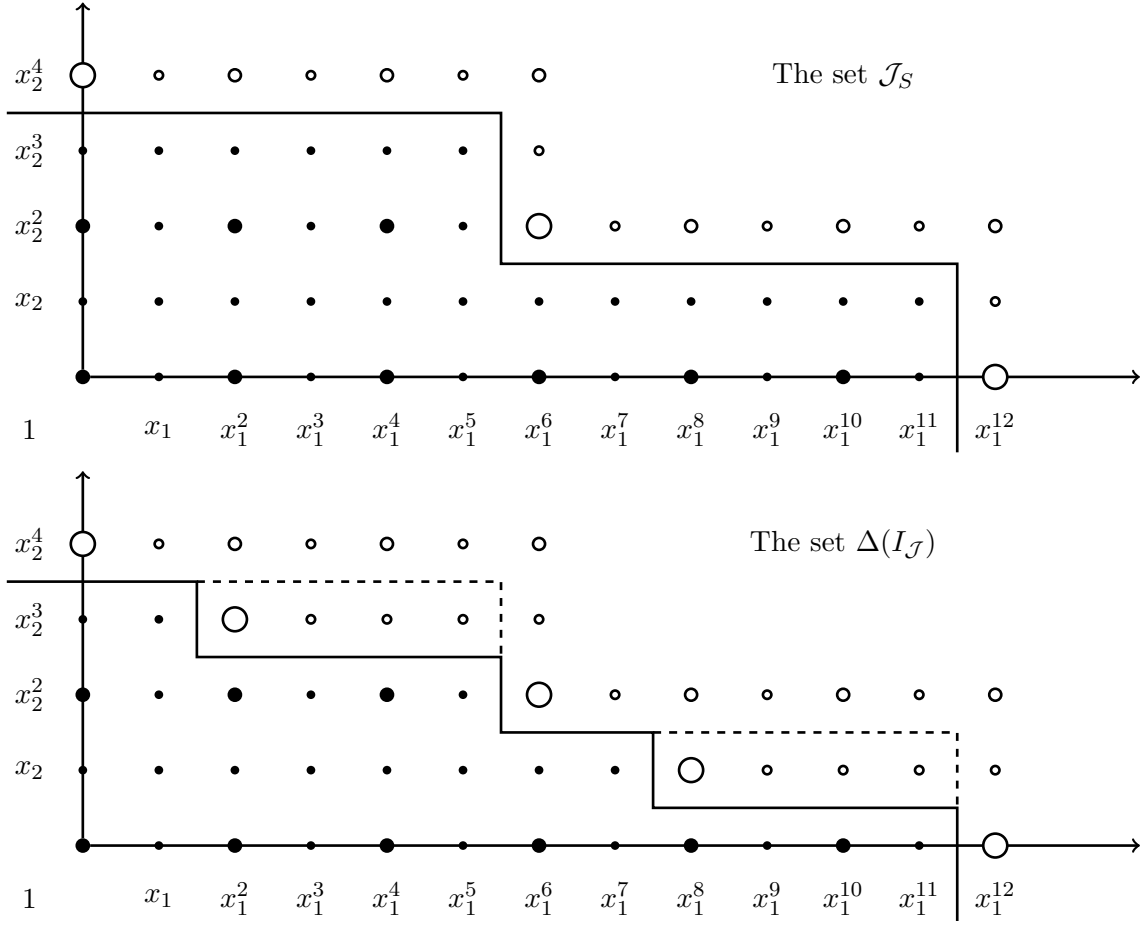

As a consequence of this interpretation, we may deduce the following useful fact:

\begin{lemma} \label{lemma auxiliary applications footprint 2}
Assume that the finite set $ \mathcal{J} \subseteq \mathbb{N}^m $ is decreasing and $ \mathbf{x}^\mathbf{i} = {\rm LM}(F(\mathbf{x})) $ with respect to some monomial ordering, for some polynomial $ F(\mathbf{x}) \in \mathbb{F}[\mathbf{x}] $. If $ \mathbf{i} \in \mathcal{J}_S $, then it holds that
\begin{equation}
\# \Delta(\langle F(\mathbf{x}) \rangle_\mathcal{J}) < \# S \cdot \# \mathcal{J}.
\label{eq suficient condigion not max footprint}
\end{equation}
\end{lemma}

We conclude with a simple description of $ \mathcal{J}_S $ in the cases of multiindices bounded by weighted orders and multiindices bounded on each coordinate separately, which follow by straightforward calculations:

\begin{remark}
Given a vector of positive weights $ \mathbf{w} = (w_1, w_2, \ldots, w_m) \in \mathbb{N}_+^m $, a positive integer $ r \in \mathbb{N}_+ $, and $ \mathcal{J} = \{ \mathbf{r} \in \mathbb{N}^m : \mid \mathbf{r} \mid_\mathbf{w} < r \} $, it holds that
$$ \mathcal{J}_S = \left\lbrace (i_1, i_2, \ldots, i_m) \in \mathbb{N}^m : \sum_{j=1}^m \left\lfloor \frac{i_j}{\# S_j} \right\rfloor w_j < r \right\rbrace . $$

On the other hand, given $ (r_1, r_2, \ldots, r_m) \in \mathbb{N}_+^m $ and $ \mathcal{J} = \left\lbrace (i_1, i_2, \ldots, i_m) \in \mathbb{N}^m : \right. $ $ i_j < r_j, $ $ \left. j = 1,2, \ldots, m \right\rbrace $, it holds that
$$ \mathcal{J}_S = \left\lbrace (i_1, i_2, \ldots, i_m) \in \mathbb{N}^m : i_j < r_j \# S_j, j = 1,2, \ldots, m \right\rbrace . $$
\end{remark}

\subsection{Sharpness and equality conditions}

To conclude the section, we study the sharpness of the bound (\ref{general bound}). We will give sufficient and necessary conditions on the ideal $ I $ for (\ref{general bound}) to be an equality, and we will see that (\ref{general bound}) is the sharpest bound that can be obtained as a strictly increasing function of the size of the footprint $ \Delta(I_\mathcal{J}) $.

We start by defining the ideal associated to a set of points and a set of multiindices.

\begin{definition} \label{def ideal of a set with multi}
Given $ \mathcal{V} \subseteq \mathbb{F}^m $, we define
$$ I(\mathcal{V}; \mathcal{J}) = \left\lbrace F(\mathbf{x}) \in \mathbb{F}[\mathbf{x}] : F^{(\mathbf{i})}(\mathbf{a}) = 0, \forall \mathbf{a} \in \mathcal{V}, \forall \mathbf{i} \in \mathcal{J} \right\rbrace . $$
\end{definition}

In the next proposition we show that this set is indeed an ideal and gather other properties similar to those of ideals and algebraic sets in algebraic geometry.

\begin{proposition} \label{prop properties of ideals of zeros}
Given a set of points $ \mathcal{V} \subseteq \mathbb{F}^m $, the set $ I(\mathcal{V}; \mathcal{J}) $ in the previous definition is an ideal in $ \mathbb{F}[\mathbf{x}] $. Moreover, the following properties hold:
\begin{enumerate}
\item
$ I \subseteq I(\mathcal{V}_\mathcal{J}(I); \mathcal{J}) $.
\item
$ \mathcal{V} \subseteq \mathcal{V}_\mathcal{J}(I(\mathcal{V}; \mathcal{J})) $.
\item
$ I = I(\mathcal{V}_\mathcal{J}(I); \mathcal{J}) $ if, and only if, $ I = I(\mathcal{W}; \mathcal{J}) $ for some set $ \mathcal{W} \subseteq \mathbb{F}^m $.
\item
$ \mathcal{V} = \mathcal{V}_\mathcal{J}(I(\mathcal{V}; \mathcal{J})) $ if, and only if, $ \mathcal{V} = \mathcal{V}_\mathcal{J}(K) $, for some ideal $ K \subseteq \mathbb{F}[\mathbf{x}] $.
\end{enumerate}
\end{proposition}
\begin{proof}
The fact that $ I(\mathcal{V}; \mathcal{J}) $ is an ideal follows from the Leibniz formula (Lemma \ref{lemma leibniz formula}) and the fact that $ \mathcal{J} $ is decreasing. The properties in items 1, 2, 3, and 4 follow as in classical algebraic geometry and are left to the reader. 
\end{proof}

The following is the main result of the subsection:

\begin{theorem} \label{th equality conditions}
Fixing a monomial ordering, the bound (\ref{general bound}) is an equality if, and only if, 
\begin{equation}
I_\mathcal{J} = I \left( \mathcal{V}_\mathcal{J}(I); \mathcal{J} \right).
\end{equation}
In particular, for any choice of decreasing finite set $ \mathcal{J} \subseteq \mathbb{N}^m $ and a finite set of points $ \mathcal{V} \subseteq \mathbb{F}^m $, there exists an ideal, $ I = I(\mathcal{V}; \mathcal{J}) $, satisfying equality in (\ref{general bound}).
\end{theorem}
\begin{proof}
With notation as in the proof of Theorem \ref{th footprint bound}, the evaluation map $ {\rm Ev} : \mathbb{F}[\mathbf{x}] \longrightarrow \mathbb{F}^{\# \mathcal{V}_\mathcal{J}(I) \cdot \# \mathcal{J}} $ from Definition \ref{def evaluation map} is $ \mathbb{F} $-linear and surjective by Lemma \ref{lemma hermite}. By definition, its kernel is
$$ {\rm Ker}({\rm Ev}) = I(\mathcal{V}_\mathcal{J}(I); \mathcal{J}). $$

On the other hand, we saw in the proof of Theorem \ref{th footprint bound} that $ I_\mathcal{J} \subseteq {\rm Ker}({\rm Ev}) $. This means that the evaluation map 
$$ {\rm Ev} : \mathbb{F}[\mathbf{x}] / I_\mathcal{J} \longrightarrow \mathbb{F}^{\# \mathcal{V}_\mathcal{J}(I) \cdot \# \mathcal{J}} $$
is an isomorphism if, and only if, $ I_\mathcal{J} = I(\mathcal{V}_\mathcal{J}(I); \mathcal{J}) $. 

Finally, the fact that this evaluation map is an isomorphism is equivalent to (\ref{general bound}) being an equality, by the proof of Theorem \ref{th footprint bound}. Together with Proposition \ref{prop properties of ideals of zeros} and the fact that $ I = I_\mathcal{J} $ if $ I = I(\mathcal{V}; \mathcal{J}) $ by the proof of Theorem \ref{th footprint bound}, the theorem follows.
\end{proof}

Thanks to this result, we may establish that the bound (\ref{general bound}) is the sharpest bound that is a strictly increasing function of the size of the footprint $ \Delta(I_\mathcal{J}) $, in the following sense: If equality holds for such a bound, then it holds in (\ref{general bound}).

\begin{corollary}
Let $ f : \mathbb{N} \longrightarrow \mathbb{R} $ be a strictly increasing function, and assume that
\begin{equation}
\# \mathcal{V}_\mathcal{J}(I) \leq f(\# \Delta (I_\mathcal{J})),
\label{eq alternative footprint bounds}
\end{equation}
for all ideals $ I \subseteq \mathbb{F}[\mathbf{x}] $. If equality holds in (\ref{eq alternative footprint bounds}) for a given ideal $ I \subseteq \mathbb{F}[\mathbf{x}] $, then equality holds in (\ref{general bound}) for such ideal.
\end{corollary}
\begin{proof}
First we have that $ I_\mathcal{J} \subseteq I(\mathcal{V}_\mathcal{J}(I); \mathcal{J}) $ as we saw in the proof of the previous theorem. Hence the reverse inclusion holds for their footprints and thus
\begin{equation}
f \left( \# \Delta \left( I \left( \mathcal{V}_\mathcal{J}(I); \mathcal{J} \right) \right) \right) \leq f(\# \Delta (I_\mathcal{J})). 
\label{eq footprint sharper proof 1}
\end{equation}
Now, since $ \mathcal{V}_\mathcal{J}(I) = \mathcal{V}_\mathcal{J}(I(\mathcal{V}_\mathcal{J}(I); \mathcal{J})) $ by Proposition \ref{prop properties of ideals of zeros}, and equality holds in (\ref{eq alternative footprint bounds}) for $ I $, we have that
\begin{equation}
f(\# \Delta (I_\mathcal{J})) = \# \mathcal{V}_\mathcal{J}(I) = \# \mathcal{V}_\mathcal{J}(I(\mathcal{V}_\mathcal{J}(I); \mathcal{J})) \leq f(\# \Delta (I(\mathcal{V}_\mathcal{J}(I); \mathcal{J}))). 
\label{eq footprint sharper proof 2}
\end{equation}
Combining (\ref{eq footprint sharper proof 1}) and (\ref{eq footprint sharper proof 2}), and using that $ f $ is strictly increasing, we conclude that 
$$ \# \Delta (I(\mathcal{V}_\mathcal{J}(I); \mathcal{J}))) = \# \Delta (I_\mathcal{J}), $$
which implies that equality holds in (\ref{general bound}) for $ I $ by Theorem \ref{th equality conditions}, and we are done.
\end{proof}

\section{Applications of the footprint bound for consecutive derivatives}

In this section, we present a brief collection of applications of Theorem \ref{th footprint bound}, which are extensions to consecutive derivatives of well-known important results from the literature. Throughout the section, we will fix again finite sets $ S_1, S_2, \ldots, S_m \subseteq \mathbb{F} $ and $ S = S_1 \times S_2 \times \cdots \times S_m $.

\subsection{Alon's combinatorial Nullstellensatz}

The combinatorial Nullstellensatz is a non-vanishing theorem by Alon \cite[Theorem 1.2]{alon} with many applications in combinatorics. It has been extended to non-vanishing theorems for standard multiplicities in \cite[Corollary 3.2]{ball} and for multisets (sets with multiplicities) in \cite[Theorem 6]{nulmultisets}. 

In this subsection, we establish and prove a combinatorial Nullstellensatz for consecutive derivatives and derive the well-known particular cases as corollaries. The formulation in \cite[Theorem 1.1]{alon} is equivalent in essence. We will extend that result in the next subsection in terms of Gr{\"o}bner bases.

\begin{theorem} \label{th combinatorial nullstellensatz general}
Let $ \mathcal{J} \subseteq \mathbb{N}^m $ be a decreasing finite set, let $ F(\mathbf{x}) \in \mathbb{F}[\mathbf{x}] $ be a non-zero polynomial, and let $ \mathbf{x}^\mathbf{i} = {\rm LM}(F(\mathbf{x})) $ for some monomial ordering. If $ \mathbf{i} \in \mathcal{J}_S $, then there exist $ \mathbf{s} \in S $ and $ \mathbf{j} \in \mathcal{J} $ such that 
$$ F^{(\mathbf{j})}(\mathbf{s}) \neq 0. $$
\end{theorem}
\begin{proof}
By Lemma \ref{lemma auxiliary applications footprint 2}, the assumptions imply that
$$ \# \Delta(\langle F(\mathbf{x}) \rangle_\mathcal{J}) < \# S \cdot \# \mathcal{J}. $$
On the other hand, Theorem \ref{th footprint bound} implies that
$$ \# \mathcal{V}_\mathcal{J}(F(\mathbf{x})) \cdot \# \mathcal{J} \leq \# \Delta(\langle F(\mathbf{x}) \rangle_\mathcal{J}). $$
Therefore not all points in $ S $ are zeros of $ F(\mathbf{x}) $ of multiplicity at least $ \mathcal{J} $, and the result follows.
\end{proof}

We now derive the original theorem \cite[Theorem 1.2]{alon}. This constitutes an alternative proof. See also \cite{shortproof} for another recent short proof.

\begin{corollary} [\textbf{\cite{alon}}]
Let $ F(\mathbf{x}) \in \mathbb{F}[\mathbf{x}] $. Assume that the coefficient of $ \mathbf{x}^\mathbf{i} $ in $ F(\mathbf{x}) $ is not zero and $ {\rm deg}(F(\mathbf{x})) = \mid \mathbf{i} \mid $. If $ \# S_j > i_j $ for all $ j = 1,2, \ldots, m $, then there exist $ s_1 \in S_1 $, $ s_2 \in S_2 $, $ \ldots $, $ s_m \in S_m $, such that
$$ F(s_1, s_2, \ldots, s_m) \neq 0. $$
\end{corollary}
\begin{proof}
First, there exists a graded monomial ordering such that $ \mathbf{x}^{\mathbf{i}} = {\rm LM}(F(\mathbf{x})) $ since $ {\rm deg}(F(\mathbf{x})) = \mid \mathbf{i} \mid $. Now, the assumption implies that 
$$ \mathbf{i} \nsucceq (r_1 \#S_1, r_2 \# S_2, \ldots, r_m \# S_m), $$
for all $ \mathbf{r} = (r_1, r_2, \ldots, r_m) $ such that $ r_j = 1 $ for some $ j $, and the rest are zero. These are in fact all minimal multiindices not in $ \mathcal{J} = \{ \mathbf{0} \} $. Thus the result follows from the previous theorem.
\end{proof}

The next consequence is a combinatorial Nullstellensatz for weighted multiplicities, where the particular case $ w_1 = w_2 = \ldots = w_m = 1 $ coincides with \cite[Corollary 3.2]{ball} (recall the definition of weighted degree from Definition \ref{def weighted degree}):

\begin{corollary} \label{theorem comb nulls weighted}
Let $ F(\mathbf{x}) \in \mathbb{F}[\mathbf{x}] $, let $ \mathbf{w} = (w_1, w_2, \ldots, w_m) \in \mathbb{N}_+^m $ and let $ r \in \mathbb{N}_+ $. Assume that the coefficient of $ \mathbf{x}^\mathbf{i} $ in $ F(\mathbf{x}) $ is not zero and $ {\rm deg}_{\mathbf{w}}(F(\mathbf{x})) = \mid \mathbf{i} \mid_{\mathbf{w}} $. 

Assume also that, for all $ \mathbf{r} = (r_1, r_2, \ldots, r_m) $ with $ \mid \mathbf{r} \mid_{\mathbf{w}} \geq r $, there exists a $ j $ such that $ r_j \# S_j > i_j $. Then there exist $ s_1 \in S_1 $, $ s_2 \in S_2 $, $ \ldots $, $ s_m \in S_m $, and some $ \mathbf{j} \in \mathbb{N}^m $ with $ \mid \mathbf{j} \mid_{\mathbf{w}} < r $, such that
$$ F^{(\mathbf{j})}(s_1, s_2, \ldots, s_m) \neq 0. $$
\end{corollary}
\begin{proof}
It follows from Theorem \ref{th combinatorial nullstellensatz general} as the previous corollary.
\end{proof}

We conclude with a combinatorial Nullstellensatz for multiindices bounded on each coordinate separately:

\begin{corollary}
Let $ F(\mathbf{x}) \in \mathbb{F}[\mathbf{x}] $, let $ (r_1, r_2, \ldots, r_m) \in \mathbb{N}_+^m $, and assume that $ \mathbf{x}^\mathbf{i} = $ $ {\rm LM}(F(\mathbf{x})) $, $ \mathbf{i} = (i_1, i_2, \ldots, i_m) $, for some monomial ordering and $ i_j < r_j \# S_j $, for all $ j = 1,2, \ldots, m $. There exist $ s_1 \in S_1 $, $ s_2 \in S_2 $, $ \ldots $, $ s_m \in S_m $, and some $ \mathbf{j} = (j_1, j_2, \ldots, j_m) \in \mathbb{N}^m $ with $ j_k < r_k $, for all $ k = 1,2, \ldots, m $, such that
$$ F^{(\mathbf{j})}(s_1, s_2, \ldots, s_m) \neq 0. $$
\end{corollary}

\subsection{Gr{\"o}bner bases of ideals of zeros in a grid}

An equivalent but more refined consequence is obtaining a Gr{\"o}bner basis for ideals $ I(S; \mathcal{J}) $ associated to the whole grid $ S $ and to a decreasing finite set of multiindices (recall Definition \ref{def ideal of a set with multi}). This result is also usually referred to as combinatorial Nullstellensatz in many works in the literature (see \cite[Theorem 1.1]{alon}, \cite[Theorem 3.1]{ball} and \cite[Theorem 1]{nulmultisets}). We briefly recall the notion of Gr{\"o}bner basis. We will also make repeated use of the Euclidean division on the multivariate polynomial ring and its properties. See \cite[Chapter 2]{clo1} for more details. 

\begin{definition}[\textbf{Gr{\"o}bner bases}]
Given a monomial ordering $ \preceq_m $ and an ideal $ I \subseteq \mathbb{F}[\mathbf{x}] $, we say that a finite family of polynomials $ \mathcal{F} \subseteq I $ is a Gr{\"o}bner basis of $ I $ with respect to $ \preceq_m $ if 
$$ \left\langle {\rm LM}_{\preceq_m}(I) \right\rangle = \left\langle {\rm LM}_{\preceq_m}(\mathcal{F}) \right\rangle. $$
Moreover, we say that $ \mathcal{F} $ is reduced if, for any two distinct $ F(\mathbf{x}), G(\mathbf{x}) \in \mathcal{F} $, it holds that $ {\rm LM}_{\preceq_m}(F(\mathbf{x})) $ does not divide any monomial in $ G(\mathbf{x}) $.
\end{definition}

Recall that a Gr{\"o}bner basis of an ideal generates it as an ideal. To obtain reduced Gr{\"o}bner bases, we need a way to minimally generate decreasing finite sets in $ \mathbb{N}^m $, which is given by the following object:

\begin{definition} \label{def generators decreasing finite set}
For any decreasing finite set $ \mathcal{J} \subseteq \mathbb{N}^m $, we define
$$ \mathcal{B}_\mathcal{J} = \{ \mathbf{i} \notin \mathcal{J} : \mathbf{j} \notin \mathcal{J} \textrm{ and } \mathbf{j} \preceq \mathbf{i} \Longrightarrow \mathbf{i} = \mathbf{j} \}. $$
\end{definition}

The main result of this subsection is the following:

\begin{theorem} \label{th groebner general}
For any decreasing finite set $ \mathcal{J} \subseteq \mathbb{N}^m $, the family
$$ \mathcal{F} = \left\lbrace \prod_{j=1}^m G_j(x_j)^{r_j} : (r_1, r_2, \ldots, r_m) \in \mathcal{B}_\mathcal{J} \right\rbrace $$
is a reduced Gr{\"o}bner basis of the ideal $ I(S; \mathcal{J}) $ with respect to any monomial ordering. In particular, for any $ F(\mathbf{x}) \in I(S; \mathcal{J}) $, there exist polynomials $ H_\mathbf{r}(\mathbf{x}) \in \mathbb{F}[\mathbf{x}] $ such that 
$$ \deg(H_\mathbf{r}(\mathbf{x})) + \sum_{j=1}^m r_j \deg(G_j(x_j)) \leq \deg(F(\mathbf{x})), $$
for $ \mathbf{r} = (r_1, r_2, \ldots, r_m) \in \mathcal{B}_\mathcal{J} $, and
$$ F(\mathbf{x}) = \sum_{\mathbf{r} \in \mathcal{B}_\mathcal{J}} \left( H_\mathbf{r}(\mathbf{x}) \prod_{j=1}^m G_j(x_j)^{r_j} \right). $$
\end{theorem}
\begin{proof}
It suffices to prove that, if $ F(\mathbf{x}) \in I(S; \mathcal{J}) $ and we divide it by the family $ \mathcal{F} $ (in an arbitrary order), then the remainder must be the zero polynomial.

Performing such division, we obtain $ F(\mathbf{x}) = G(\mathbf{x}) + R(\mathbf{x}) $, where $ R(\mathbf{x}) $ is the remainder of the division and $ G(\mathbf{x}) \in I(S; \mathcal{J}) $. Assume that $ R(\mathbf{x}) \neq 0 $ and let $ \mathbf{x}^\mathbf{i} $ be the leading monomial of $ R(\mathbf{x}) $ with respect to the chosen monomial ordering. Since no leading monomial of the polynomials in $ \mathcal{F} $ divides $ \mathbf{x}^\mathbf{i} $, we conclude that
$$ \mathbf{i} \nsucceq (r_1 \#S_1, r_2 \#S_2, \ldots, r_m \# S_m), $$
for all minimal $ \mathbf{r} = (r_1, r_2, \ldots, r_m) \notin \mathcal{J} $, that is, for all $ \mathbf{r} \in \mathcal{B}_\mathcal{J} $. Thus by Theorem \ref{th combinatorial nullstellensatz general}, we conclude that not all points in $ S $ are zeros of $ R(\mathbf{x}) $ of multiplicity at least $ \mathcal{J} $, which is absurd since $ R(\mathbf{x}) = F(\mathbf{x}) - G(\mathbf{x}) \in I(S; \mathcal{J}) $, and we are done.

The fact that $ \mathcal{F} $ is reduced follows from observing that the multiindices $ \mathbf{r} \in \mathcal{B}_\mathcal{J} $ are minimal among those not in $ \mathcal{J} $. The last part of the theorem follows by performing the Euclidean division.
\end{proof}

The following particular case is \cite[Theorem 1.1]{alon}:

\begin{corollary} [\textbf{\cite{alon}}]
If $ F(\mathbf{x}) \in \mathbb{F}[\mathbf{x}] $ vanishes at all points in $ S $, then there exist polynomials $ H_j(\mathbf{x}) \in \mathbb{F}[\mathbf{x}] $ such that $ \deg(H_j(\mathbf{x})) + \deg(G_j(x_j)) \leq \deg(F(\mathbf{x})) $, for $ j = 1,2, \ldots, m $, and
$$ F(\mathbf{x}) = \sum_{j=1}^m H_j(\mathbf{x}) G_j(x_j). $$
\end{corollary}

To study the case of weighted multiplicities, we observe the following:

\begin{remark}
Given a vector of positive weights $ \mathbf{w} = (w_1, w_2, \ldots, w_m) \in \mathbb{N}_+^m $, a positive integer $ r \in \mathbb{N}_+ $, and the set $ \mathcal{J} = \{ \mathbf{i} \in \mathbb{N}^m : \mid \mathbf{i} \mid_\mathbf{w} < r \} $, it holds that $ \mathcal{B}_\mathcal{J} = \mathcal{B}_\mathbf{w} $, where
$$ \mathcal{B}_\mathbf{w} = \left\lbrace (i_1, i_2, \ldots, i_m) \in \mathbb{N}^m : r \leq \sum_{j=1}^m i_j w_j < r + \min \left\lbrace w_j : i_j \neq 0 \right\rbrace   \right\rbrace . $$
\end{remark}

We then obtain the next consequence, where the particular case $ w_1 = w_2 = \ldots = w_m = 1 $ coincides with \cite[Theorem 3.1]{ball}.

\begin{corollary} \label{corollary groebner weighted}
Given a vector of positive weights $ \mathbf{w} = (w_1, w_2, \ldots, w_m) \in \mathbb{N}_+^m $ and a positive integer $ r \in \mathbb{N}_+ $, if $ F(\mathbf{x}) \in \mathbb{F}[\mathbf{x}] $ vanishes at all points in $ S $ with weighted multiplicity at least $ r $, then there exist polynomials $ H_\mathbf{r}(\mathbf{x}) \in \mathbb{F}[\mathbf{x}] $ such that $ \deg(H_\mathbf{r}(\mathbf{x})) + \sum_{j=1}^m r_j \deg(G_j(x_j)) \leq \deg(F(\mathbf{x})) $, for all $ \mathbf{r} = (r_1, r_2, \ldots, r_m) \in \mathcal{B}_\mathbf{w} $, and
$$ F(\mathbf{x}) = \sum_{\mathbf{r} \in \mathcal{B}_\mathbf{w}} \left( H_\mathbf{r}(\mathbf{x}) \prod_{j=1}^m G_j(x_j)^{r_j} \right). $$
\end{corollary}

We conclude with the case of multiindices bounded on each coordinate separately:

\begin{corollary} \label{corollary groebner coordinate}
Given a vector $ (r_1, r_2, \ldots, r_m) \in \mathbb{N}_+^m $, if $ F(\mathbf{x}) \in \mathbb{F}[\mathbf{x}] $ is such that $ F^{(\mathbf{j})}(\mathbf{s}) = 0 $, for all $ \mathbf{s} \in S $ and all $ \mathbf{j} = (j_1, j_2, \ldots, j_m) \in \mathbb{N}^m $ satisfying $ j_k < r_k $, for all $ k = 1,2, \ldots, m $, then there exist polynomials $ H_j(\mathbf{x}) \in \mathbb{F}[\mathbf{x}] $ such that $ \deg(H_j(\mathbf{x})) + r_j \deg(G_j(x_j)) \leq \deg(F(\mathbf{x})) $, for all $ j = 1,2, \ldots, m $, and
$$ F(\mathbf{x}) = \sum_{j = 1}^m H_j(\mathbf{x}) G_j(x_j)^{r_j}. $$
\end{corollary}
\begin{proof}
It follows from Theorem \ref{th groebner general} observing that, if $ \mathcal{J} = \left\lbrace (j_1, j_2, \ldots, j_m) \in \mathbb{N}^m : j_k < r_k, \right. $ $ \left. k = 1,2, \ldots, m \right\rbrace $, then
$$ \mathcal{B}_\mathcal{J} = \left\lbrace r_j \mathbf{e}_j \in \mathbb{N}^m : j = 1,2, \ldots, m \right\rbrace, $$
where $ \mathbf{e}_1, \mathbf{e}_2, \ldots, \mathbf{e}_m \in \mathbb{N}^m $ are the vectors in the canonical basis.
\end{proof}

\subsection{Hermite interpolation over grids with consecutive derivatives}

In the appendix we show that the evaluation map (Definition \ref{def evaluation map}) is surjective. This has been used to prove Theorem \ref{th footprint bound}. In this subsection, we see that the combinatorial Nullstellensatz (Theorem \ref{th combinatorial nullstellensatz general}) implies that the evaluation map over the whole grid $ S $, with consecutive derivatives, is an isomorphism when taking an appropriate domain. More concretely, we show the existence and uniqueness of Hermite interpolating polynomials over $ S $ with derivatives in $ \mathcal{J} $ when choosing monomials in $ \mathcal{J}_S $. Finding appropriate sets of points, derivatives and polynomials to guarantee existence and uniqueness of Hermite interpolating polynomials has been extensively studied \cite{gasca, kopparty-multiplicity, lorentz}. The next result is new to the best of our knowledge:

\begin{theorem} \label{th uniqueness hermite}
Given a decreasing finite set $ \mathcal{J} \subseteq \mathbb{N}^m $, the evaluation map in Definition \ref{def evaluation map} for the finite set $ S = S_1 \times S_2 \times \cdots \times S_m $, defined as
$$ {\rm Ev} : \left\langle \mathcal{J}_S \right\rangle_\mathbb{F} \longrightarrow \mathbb{F}^{\# S \cdot \# \mathcal{J}}, $$
is a vector space isomorphism. In other words, for all $ b_{\mathbf{j}, \mathbf{a}} \in \mathbb{F} $, where $ \mathbf{j} \in \mathcal{J} $ and $ \mathbf{a} \in S $, there exists a unique polynomial of the form 
$$ F(\mathbf{x}) = \sum_{\mathbf{i} \in \mathcal{J}_S} F_\mathbf{i} \mathbf{x}^\mathbf{i} \in \mathbb{F}[\mathbf{x}], $$
where $ F_\mathbf{i} \in \mathbb{F} $ for all $ \mathbf{i} \in \mathcal{J}_S $, such that $ F^{(\mathbf{j})}(\mathbf{a}) = b_{\mathbf{j}, \mathbf{a}} $, for all $ \mathbf{j} \in \mathcal{J} $ and all $ \mathbf{a} \in S $.
\end{theorem}
\begin{proof}
The map is one to one by Theorem \ref{th combinatorial nullstellensatz general}, and both vector spaces have the same dimension over $ \mathbb{F} $ by Lemma \ref{lemma auxiliary applications footprint 1}, hence the map is a vector space isomorphism.
\end{proof}

\begin{remark}
Observe that we may similarly prove that the following two maps are vector space isomorphisms:
$$ \left\langle \mathcal{J}_S \right\rangle_\mathbb{F} \stackrel{\rho}{\longrightarrow} \mathbb{F}[\mathbf{x}] / I(S; \mathcal{J}) \stackrel{{\rm Ev}}{\longrightarrow} \mathbb{F}^{\# S \cdot \# \mathcal{J}}, $$
where $ \rho $ is the projection to the quotient ring. We may then extend the notion of \textit{reduction} of a polynomial as follows (see \cite[Section 3.1]{clark} and \cite[Section 6.3]{gasca}, for instance): Given $ F(\mathbf{x}) \in \mathbb{F}[\mathbf{x}] $, we define its reduction over the set $ S $ with derivatives in $ \mathcal{J} $ as
$$ G(\mathbf{x}) = \rho^{-1} \left( F(\mathbf{x}) + I(S; \mathcal{J}) \right) . $$
\end{remark}

As an immediate consequence, we obtain the following result on Hermite interpolation with weighted multiplicities:

\begin{corollary} \label{cor uniqueness hermite weighted}
For every vector of positive weights $ \mathbf{w} = (w_1, w_2, \ldots, w_m) \in \mathbb{N}_+^m $, every positive integer $ r \in \mathbb{N}_+ $, and elements $ b_{\mathbf{j}, \mathbf{a}} \in \mathbb{F} $, for $ \mathbf{j} \in \mathbb{N}^m $ with $ \mid \mathbf{j} \mid_\mathbf{w} < r $ and for $ \mathbf{a} \in S $, there exists a unique polynomial of the form 
$$ F(\mathbf{x}) = \sum_{\mathbf{i} \in \mathbb{N}^m} F_\mathbf{i} \mathbf{x}^\mathbf{i}, $$
where $ F_\mathbf{i} \in \mathbb{F} $ for all $ \mathbf{i} = (i_1, i_2, \ldots, i_m) \in \mathbb{N}^m $, and $ F_\mathbf{i} = 0 $ whenever 
$$ \sum_{j=1}^m \left\lfloor \frac{i_j}{\# S_j} \right\rfloor w_j \geq r, $$ 
such that $ F^{(\mathbf{j})}(\mathbf{a}) = b_{\mathbf{j}, \mathbf{a}} $, for all $ \mathbf{j} \in \mathbb{N}^m $ with $ \mid \mathbf{j} \mid_\mathbf{w} < r $ and all $ \mathbf{a} \in S $.
\end{corollary}

We conclude with the case of multiindices bounded on each coordinate separately:

\begin{corollary} \label{cor uniqueness hermite coordintate}
Given $ (r_1, r_2, \ldots, r_m) \in \mathbb{N}_+^m $ and given elements $ b_{\mathbf{j}, \mathbf{a}} \in \mathbb{F} $, for $ \mathbf{j} = (j_1, j_2, \ldots, j_m) \in \mathbb{N}^m $ with $ j_k < r_k $, for all $ k = 1,2, \ldots, m $, and for $ \mathbf{a} \in S $, there exists a unique polynomial of the form 
$$ F(\mathbf{x}) = \sum_{i_1 = 0}^{r_1\# S_1 - 1} \sum_{i_2 = 0}^{r_2\# S_2 - 1} \cdots \sum_{i_m = 0}^{r_m\# S_m - 1} F_\mathbf{i} \mathbf{x}^\mathbf{i}, $$
such that $ F^{(\mathbf{j})}(\mathbf{a}) = b_{\mathbf{j}, \mathbf{a}} $, for all $ \mathbf{j} = (j_1, j_2, \ldots, j_m) \in \mathbb{N}^m $ with $ j_k < r_k $, for all $ k = 1,2, \ldots, m $, and all $ \mathbf{a} \in S $.
\end{corollary}

\subsection{Evaluation codes with consecutive derivatives}

In this subsection, we extend the notion of \textit{evaluation code} from the theory of error-correcting codes (see \cite[Section 2]{weightedRM} and \cite[Section 4.1]{handbook}, for instance) to evaluation codes with consecutive derivatives. By doing so, we generalize \textit{multiplicity codes} \cite{multiplicitycodes}, which have been shown to achieve good parameters in decoding, local decoding and list decoding \cite{kopparty-multiplicity, multiplicitycodes}. We compute the dimensions of the new codes and give a lower bound on their minimum Hamming distance.

\begin{definition} \label{def evaluation codes with derivatives}
Given a decreasing finite set $ \mathcal{J} \subseteq \mathbb{N}^m $ and a set of monomials $ \mathcal{M} \subseteq \mathcal{J}_S $, we define the $ \mathbb{F} $-linear code (that is, the $ \mathbb{F} $-linear vector space)
$$ \mathcal{C}(S, \mathcal{M}, \mathcal{J}) = {\rm Ev} \left( \left\langle \mathcal{M} \right\rangle _\mathbb{F} \right) \subseteq \mathbb{F}^{\# S \cdot \# \mathcal{J}}, $$
where $ {\rm Ev} $ is the evaluation map from Definition \ref{def evaluation map}.
\end{definition}

As in \cite{multiplicitycodes}, we will consider these codes over the alphabet $ \mathbb{F}^{\# \mathcal{J}} $, that is, each evaluation $ \left( F^{(\mathbf{i})} \left( \mathbf{a} \right) \right)_{\mathbf{i} \in \mathcal{J}} \in \mathbb{F}^{\# \mathcal{J}} $, for $ \mathbf{a} \in S $, constitutes one symbol of the alphabet. Thus each codeword has length $ \# S $ over this alphabet. This leads to the following definition of minimum Hamming distance of an $ \mathbb{F} $-linear code:

\begin{definition}
Given an $ \mathbb{F} $-linear code $ \mathcal{C} \subseteq \left( \mathbb{F}^{\# \mathcal{J}} \right) ^{\# S} $, we define its minimum Hamming distance as
$$ d_H(\mathcal{C}) = \min \left\lbrace {\rm wt_H}(\mathbf{c}) : \mathbf{c} \in \mathcal{C}, \mathbf{c} \neq \mathbf{0} \right\rbrace, $$
where, for any $ \mathbf{c} \in \left( \mathbb{F}^{\# \mathcal{J}} \right) ^{\# S} $, $ {\rm wt_H}(\mathbf{c}) $ denotes the number of its non-zero components over the alphabet $ \mathbb{F}^{\# \mathcal{J}} $.
\end{definition}

As a consequence of Theorem \ref{th uniqueness hermite}, we may exactly compute the dimensions of the codes in Definition \ref{def evaluation codes with derivatives} and give a lower bound on their minimum Hamming distance:

\begin{corollary}
The code in Definition \ref{def evaluation codes with derivatives} satisfies that
$$ \dim_\mathbb{F}(\mathcal{C}(S, \mathcal{M}, \mathcal{J})) = \# \mathcal{M}, \quad \textrm{and} $$
$$ d_H(\mathcal{C}(S, \mathcal{M}, \mathcal{J})) \geq \left\lceil \frac{\min \left\lbrace \# \Delta( \left\langle F(\mathbf{x}) \right\rangle _\mathcal{J}) : F(\mathbf{x}) \in \left\langle \mathcal{M} \right\rangle _\mathbb{F} \right\rbrace}{\# \mathcal{J}} \right\rceil . $$
\end{corollary}

\begin{remark} \label{remark weighted multiplicity codes}
Given a vector of positive weights $ \mathbf{w} = (w_1, w_2, \ldots, w_m) \in \mathbb{N}_+^m $, a positive integer $ r \in \mathbb{N}_+ $, and a set of monomials 
$$ \mathcal{M} \subseteq \left\lbrace x_1^{i_1} x_2^{i_2} \cdots x_m^{i_m} : \sum_{j=1}^m \left\lfloor \frac{i_j}{\# S_j} \right\rfloor w_j < r \right\rbrace, $$
we may define, as a particular case of the codes in Definition \ref{def evaluation codes with derivatives}, the corresponding weighted multiplicity code as the $ \mathbb{F} $-linear code 
$$ \mathcal{C}(S, \mathcal{M}, \mathbf{w}, r) = {\rm Ev} \left( \left\langle \mathcal{M} \right\rangle _\mathbb{F} \right) \subseteq \left( \mathbb{F}^{{\rm B}(\mathbf{w} ; r)} \right)^{\# S} . $$
Observe that weighted multiplicity codes contain as particular cases classical Reed-Muller codes (see \cite[Section 13.2]{pless}), by choosing $ \mathbf{w} = (r,r, \ldots, r) $ for a given $ r \in \mathbb{N}_+ $, and classical multiplicity codes \cite{multiplicitycodes} by choosing $ \mathbf{w} = (1,1, \ldots, 1) $ and an arbitrary $ r \in \mathbb{N}_+ $. Therefore, choices of $ \mathbf{w} \in \mathbb{N}^m $ such that $ 1 \leq w_i \leq r $, for $ i =1,2, \ldots, m $, give codes with the same length but intermediate alphabet sizes between those of Reed-Muller and multiplicity codes. This has the extra flexibility (see \cite[Section 1.2]{multiplicitycodes}) of choosing alphabets of sizes $ \# \left( \mathbb{F}^{{\rm B}(\mathbf{w}; r)} \right) $ (whenever $ \mathbb{F} $ is finite), where
$$ 1 \leq {\rm B}(\mathbf{w}; r) \leq \binom{m+r-1}{m}. $$
\end{remark}

\subsection{Bounds by DeMillo, Lipton, Zippel, Alon and F{\"u}redi}

In this subsection, we obtain a weaker but more concise version of the bound (\ref{general bound}) for a single polynomial, which has as particular cases the bounds by DeMillo and Lipton \cite{demillo}, Zippel \cite[Theorem 1]{zippel-first}, \cite[Proposition 3]{zippel}, and Alon and F{\"u}redi \cite[Theorem 5]{alon-furedi}. We observe that Counterexample 7.4 in \cite{grid} shows that a straightforward extension of these bounds to standard multiplicities as in (\ref{eq in intro 1}) is not possible, in contrast with the bound given by Schwartz in \cite[Lemma 1]{schwartz}, which has been already extended in \cite[Lemma 8]{extensions}.

\begin{theorem} \label{th generalization demillo lipton et al}
For any decreasing finite set $ \mathcal{J} \subseteq \mathbb{N}^m $ and any polynomial $ F(\mathbf{x}) \in \mathbb{F}[\mathbf{x}] $, if $ \mathbf{x}^\mathbf{i} = {\rm LM}(F(\mathbf{x})) \in \mathcal{J}_S $, for some monomial ordering, then it holds that
\begin{equation}
\# \left( S \setminus \mathcal{V}_\mathcal{J}(F(\mathbf{x})) \right) \# \mathcal{J} \geq \# \left\lbrace \mathbf{j} \in \mathcal{J}_S : \mathbf{j} \succeq \mathbf{i} \right\rbrace.
\label{eq bound generalization of demillo-lipton}
\end{equation}
\end{theorem}
\begin{proof}
First, from the bound (\ref{general bound}) and Lemma \ref{lemma auxiliary applications footprint 1}, we obtain that
\begin{equation}
 \# \left( S \setminus \mathcal{V}_\mathcal{J}(F(\mathbf{x})) \right) \# \mathcal{J} \geq \# S \# \mathcal{J} - \# \Delta( \left\langle F(\mathbf{x}) \right\rangle_\mathcal{J}) = \# \left( \mathcal{J}_S \setminus \Delta( \left\langle F(\mathbf{x}) \right\rangle_\mathcal{J}) \right) ,
\label{eq proof bound demillo-lipton}
\end{equation}
where we consider $ \Delta( \left\langle F(\mathbf{x}) \right\rangle_\mathcal{J}) \subseteq \mathbb{N}^m $ by abuse of notation. As explained in Subsection \ref{subsec interpretation bound}, we may lower bound $ \# \left( \mathcal{J}_S \setminus \Delta( \left\langle F(\mathbf{x}) \right\rangle_\mathcal{J}) \right) $ by the number of multiindices $ \mathbf{j} \in \mathcal{J}_S $ satisfying $ \mathbf{j} \succeq \mathbf{i} $, and we are done.
\end{proof}

The following consequence summarizes the results by DeMillo and Lipton \cite{demillo}, and Zippel \cite[Theorem 1]{zippel-first}, \cite[Proposition 3]{zippel}:

\begin{corollary}[\textbf{\cite{demillo, zippel-first, zippel}}]
Let $ F(\mathbf{x}) \in \mathbb{F}[\mathbf{x}] $ be such that its degree in the $ j $-th variable is $ d_j \in \mathbb{N} $, for $ j = 1,2, \ldots, m $. If $ d_j < \# S_j $, for $ j = 1,2, \ldots, m $, then the number of non-zeros in $ S $ of $ F(\mathbf{x}) $ is at least
$$ \prod_{j=1}^m \left( \# S_j - d_j \right). $$
\end{corollary}
\begin{proof}
The result is the particular case $ \mathcal{J} = \{ \mathbf{0} \} $ of the previous theorem using any monomial ordering and the facts that $ \mathcal{J}_S = S $ and $ i_j \leq d_j $, for $ j = 1,2, \ldots, m $. 
\end{proof}

The following is a similar bound due to Alon and F{\"u}redi \cite[Theorem 5]{alon-furedi}:

\begin{corollary}[\textbf{\cite{alon-furedi}}]
Let $ F(\mathbf{x}) \in \mathbb{F}[\mathbf{x}] $. If not all points in $ S $ are zeros of $ F(\mathbf{x}) $, then the number of its non-zeros in $ S $ is at least
$$ \min \left\lbrace \prod_{j=1}^m y_j : 1 \leq y_j \leq \# S_j, \sum_{j=1}^m y_j \geq \sum_{j=1}^m \# S_j - \deg(F(\mathbf{x})) \right\rbrace. $$
\end{corollary}
\begin{proof}
The result follows from Theorem \ref{th generalization demillo lipton et al} as in the previous corollary, taking any monomial ordering and considering $ y_j = \# S_j - i_j $, for $ j = 1,2, \ldots, m $.
\end{proof}

We omit the case of weighted multiplicities. In the next section, we will give an extension of the bound given by Schwartz in \cite[Lemma 1]{schwartz} to weighted multiplicities in the sense of (\ref{eq in intro 1}), which is stronger than the bound in Corollary \ref{corollary footprint for weighted multi} in some cases.

We conclude with the case of multiindices bounded on each coordinate separately:

\begin{corollary}
Let $ F(\mathbf{x}) \in \mathbb{F}[\mathbf{x}] $ with $ \mathbf{x}^\mathbf{i} = {\rm LM}(F(\mathbf{x})) $, $ \mathbf{i} = (i_1, i_2, \ldots, i_m) $, for some monomial ordering. If $ i_j < r_j \# S_j $, for $ j = 1,2, \ldots, m $, then the number $ N $ of elements $ \mathbf{s} \in S $ such that $ F^{(\mathbf{j})}(\mathbf{s}) \neq 0 $, for some $ \mathbf{j} = (j_1, j_2, \ldots, j_m) \in \mathbb{N}^m $ with $ j_k < r_k $, for all $ k = 1,2, \ldots, m $, satisfies
$$ N \cdot \prod_{j=1}^m r_j \geq \prod_{j=1}^m \left( r_j \# S_j - i_j \right). $$
\end{corollary}

\subsection{The Schwartz-Zippel bound on the whole grid}

In the next section, we will give an extension of bound given by Schwartz in \cite[Lemma 1]{schwartz} for weighted multiplicities that can be proven as the extensions to standard multiplicities given in \cite[Lemma 8]{extensions} and \cite[Theorem 5]{weightedRM}. In this subsection, we observe that the case where all points in $ S $ are zeros of a given weighted multiplicity follows from Corollary \ref{theorem comb nulls weighted}:

\begin{corollary} \label{corollary weak schwartz on whole grid}
Let $ F(\mathbf{x}) \in \mathbb{F}[\mathbf{x}] $, let $ \mathbf{w} = (w_1, w_2, \ldots, w_m) \in \mathbb{N}_+^m $, let $ r \in \mathbb{N}_+ $, and assume that $ s = \# S_1 = \# S_2 = \ldots = \# S_m $. If all points in $ S = S_1 \times S_2 \times \cdots \times S_m $ are zeros of $ F(\mathbf{x}) $ of weighted multiplicity at least $ r $, then
$$ r \# S \leq \deg_\mathbf{w} (F(\mathbf{x})) s^{m-1}. $$
\end{corollary}
\begin{proof}
Assume that the bound does not hold, take $ \mathbf{x}^\mathbf{i} $ such that $ \mid \mathbf{i} \mid_\mathbf{w} = \deg_\mathbf{w}(F(\mathbf{x})) $ and whose coefficient in $ F(\mathbf{x}) $ is not zero, and take a vector $ \mathbf{r} = (r_1, r_2, \ldots, r_m) \in \mathbb{N}^m $ with $ \mid \mathbf{r} \mid_\mathbf{w} \geq r $. Then
$$ s w_1 r_1 + s w_2 r_2 + \cdots + s w_m r_m \geq sr > \deg_\mathbf{w} (F(\mathbf{x})) = \mid \mathbf{i} \mid_\mathbf{w}, $$
hence there exists a $ j $ such that $ r_j \# S_j > i_j $. By Corollary \ref{theorem comb nulls weighted}, some element in $ S $ is not a zero of $ F(\mathbf{x}) $ of weighted multiplicity at least $ r $, which contradicts the assumptions and we are done.
\end{proof}

\section{The Schwartz-Zippel bound for weighted multiplicities}

As we will see in Proposition \ref{proposition zippel implies schwartz}, the bound given by Schwartz in \cite[Lemma 1]{schwartz} can be derived by those given by DeMillo and Lipton \cite{demillo}, and Zippel \cite[Theorem 1]{zippel-first}, \cite[Proposition 3]{zippel}, and is usually referred to as the Schwartz-Zippel bound. This bound has been recently extended to standard multiplicities in \cite[Lemma 8]{extensions}, and further in \cite[Theorem 5]{weightedRM}. In this section, we observe that it may be easily extended to weighted multiplicities (see Definition \ref{def weighted multiplicity}), due to the additivity of weighted order functions. We show the sharpness of this bound and compare it with the bound (\ref{general bound}) with an example, whenever it makes sense to compare both bounds.

\subsection{The bound}

\begin{theorem} \label{th S-Z bound}
Let $ \mathbf{w} = (w_1, w_2, \ldots, w_m ) \in \mathbb{N}_+^m $ be a vector of positive weights, let $ F(\mathbf{x}) \in \mathbb{F}[\mathbf{x}] $ and let $ \mathbf{x}^{\mathbf{i}} = {\rm LM}(F(\mathbf{x})) $, $ \mathbf{i} = (i_1, i_2, \ldots, i_m) $, with respect to the lexicographic ordering. It holds that
\begin{equation}
\sum_{\mathbf{a} \in S} m_{\mathbf{w}}(F(\mathbf{x}), \mathbf{a}) \leq \#S \sum_{j=1}^m \frac{i_j w_j}{\# S_j}.
\label{eq S-Z bound}
\end{equation}
\end{theorem}

When $ w_1 = w_2 = \ldots = w_m = 1 $, observe that \cite[Theorem 5]{weightedRM} is recovered from this theorem, and \cite[Lemma 8]{extensions} is recovered from the next corollary. Observe also that this corollary is stronger than Corollary \ref{corollary weak schwartz on whole grid}.

\begin{corollary}
Let $ F(\mathbf{x}) \in \mathbb{F}[\mathbf{x}] $ and $ \mathbf{w} \in \mathbb{N}_+^m $. If $ s = \# S_1 = \# S_2 = \ldots = \# S_m $, then
$$ \sum_{\mathbf{a} \in S} m_{\mathbf{w}}(F(\mathbf{x}), \mathbf{a}) \leq \deg_{\mathbf{w}}(F(\mathbf{x})) s^{m-1}. $$
\end{corollary}

To prove Theorem \ref{th S-Z bound}, we need an auxiliary lemma, whose proof can be directly translated from those of \cite[Lemma 5]{extensions} and \cite[Corollary 7]{extensions}:

\begin{lemma} \label{lemma SZ}
If $ F(\mathbf{x}) \in \mathbb{F}[\mathbf{x}] $ and $ \mathbf{a} = (a_1, a_2, \ldots, a_m) \in \mathbb{F}^m $, then
\begin{enumerate}
\item
$ m_{\mathbf{w}} \left( F^{(\mathbf{i})}(\mathbf{x}), \mathbf{a} \right) \geq m_{\mathbf{w}}(F(\mathbf{x}), \mathbf{a}) - \mid\mathbf{i}\mid_{\mathbf{w}} $, for all $ \mathbf{i} \in \mathbb{N}^m $, and
\item
$ m_{\mathbf{w}} \left( F(\mathbf{x}), \mathbf{a} \right) \leq m_{w_m}(F(a_1,a_2, \ldots, a_{m-1},x_m), a_m) $.
\end{enumerate}
\end{lemma}

We may now prove Theorem \ref{th S-Z bound}. We follow closely the steps given in the proof of \cite[Lemma 8]{extensions}.

\begin{proof}[Proof of Theorem \ref{th S-Z bound}]
We will prove the result by induction on $ m $, where the case $ m = 1 $ follows from (\ref{eq in intro 1}). Fix then $ m > 1 $. We may assume without loss of generality that $ x_1 \prec_l x_2 \prec_l \ldots \prec_l x_m $, where $ \preceq_l $ is the lexicographic ordering. Write $ \mathbf{x}^\prime = (x_1, x_2, \ldots, x_{m-1}) $. There are unique polynomials $ F_j(\mathbf{x}^\prime) \in \mathbb{F}[\mathbf{x}^\prime] $, for $ j = 1,2, \ldots, t $, such that
$$ F(\mathbf{x}) = \sum_{j=0}^t F_j(\mathbf{x}^\prime) x_m^j, $$
where $ {\rm LM}(F(\mathbf{x})) = {\rm LM}(F_t(\mathbf{x}^\prime)) x_m^t $. Let $ \mathbf{a} = (a_1, a_2, \ldots, a_m) \in S $ and write $ \mathbf{a}^\prime = (a_1, a_2, \ldots, a_{m-1}) $ and $ \mathbf{w}^\prime = (w_1, w_2, \ldots, $ $ w_{m-1}) $. Take $ \mathbf{k} \in \mathbb{N}^{m-1} $ such that $ \mid\mathbf{k}\mid_{\mathbf{w}^\prime} =  m_{\mathbf{w}^\prime}(F_t(\mathbf{x}^\prime), \mathbf{a}^\prime) $ and $ F_t^{(\mathbf{k})}(\mathbf{a}^\prime) \neq 0 $. By the previous lemma, we see that
\begin{equation*}
\begin{split}
m_{\mathbf{w}}(F(\mathbf{x}), \mathbf{a}) & \leq \mid(\mathbf{k},0)\mid_{\mathbf{w}} + m_{\mathbf{w}} \left( F^{(\mathbf{k},0)}(\mathbf{x}), \mathbf{a} \right) \\
 & \leq m_{\mathbf{w}^\prime}(F_t(\mathbf{x}^\prime), \mathbf{a}^\prime) +  m_{w_m} \left( F^{(\mathbf{k},0)}(\mathbf{a}^\prime, x_m), a_m \right) .
\end{split}
\end{equation*}
Summing these inequalities over all $ a_m \in S_m $ and applying the case $ m=1 $, we obtain that
\begin{equation*}
\sum_{a_m \in S_m} m_{\mathbf{w}}(F(\mathbf{x}), \mathbf{a}) \leq  m_{\mathbf{w}^\prime}(F_t(\mathbf{x}^\prime), \mathbf{a}^\prime) \# S_m + w_mt.
\end{equation*}
Using this last inequality, summing over $ a_i \in S_i $, for $ i = 1,2, \ldots, m-1 $, and applying the case of $ m-1 $ variables, it follows that
\begin{equation*}
\begin{split}
\sum_{\mathbf{a} \in S} m_{\mathbf{w}}(F(\mathbf{x}), \mathbf{a}) & \leq \sum_{a_1 \in S_1} \cdots \sum_{a_{m-1} \in S_{m-1}} m_{\mathbf{w}^\prime}(F_t(\mathbf{x}^\prime), \mathbf{a}^\prime) \# S_m + w_mt \frac{\# S}{\#S_m} \\
& \leq \sum_{j=1}^{m-1} w_ji_j \frac{\# S}{\# S_j} + w_mt \frac{\# S}{\#S_m},
\end{split}
\end{equation*}
and the result follows.
\end{proof}

\subsection{Sharpness of the bound}

In this subsection, we prove the sharpness of the bound (\ref{eq S-Z bound}), whose proof can be translated word by word from that of \cite[Proposition 7]{moreresults}. Therefore, we only present a sketch of the proof:

\begin{proposition}
For all finite sets $ S_1, S_2, \ldots, S_m \subseteq \mathbb{F} $, $ S = S_1 \times S_2 \times \cdots \times S_m $, all vectors of positive weights $ \mathbf{w} = (w_1, w_2, \ldots, w_m) \in \mathbb{N}_+^m $ and all $ \mathbf{i} = (i_1, i_2, \ldots, i_m) \in \mathbb{N}^m $, there exists a polynomial $ F(\mathbf{x}) \in \mathbb{F}[\mathbf{x}] $ such that $ \mathbf{x}^\mathbf{i} = {\rm LM}(F(\mathbf{x})) $ with respect to the lexicographic ordering, and such that
$$ \sum_{\mathbf{a} \in S} m_{\mathbf{w}}(F(\mathbf{x}), \mathbf{a}) = \#S \sum_{j=1}^m \frac{i_j w_j}{\# S_j}. $$
\end{proposition}
\begin{proof}[Sketch of proof]
Denote $ s_j = \# S_j $ and $ S_j = \left\lbrace a^{(j)}_1, a^{(j)}_2, \ldots, a^{(j)}_{s_j} \right\rbrace $, and choose $ r^{(j)}_k \in \mathbb{N} $ such that $ i_j = r^{(j)}_1 + r^{(j)}_2 + \cdots + r^{(j)}_{s_j} $, for $ k = 1,2, \ldots, s_j $ and $ j = 1,2, \ldots, m $. Now define
$$ F(\mathbf{x}) = \prod_{j = 1}^m \prod_{k = 1}^{s_j} \left( x_j - a^{(j)}_k \right)^{r^{(j)}_k}. $$
Now, fixing integers $ 1 \leq k_j \leq s_j $, for $ j = 1,2, \ldots, m $, translating the point $ \left( a^{(1)}_{k_1}, a^{(2)}_{k_2}, \ldots, \right. $ $ \left. a^{(m)}_{k_m} \right) $ to the origin $ \mathbf{0} $, and using the Gr{\"o}bner basis from Corollary \ref{corollary groebner weighted}, we see that
$$ m_\mathbf{w} \left( F \left( \mathbf{x} \right), \left( a^{(1)}_{k_1}, a^{(2)}_{k_2}, \ldots, a^{(m)}_{k_m} \right) \right) = r^{(1)}_{k_1} w_1 + r^{(2)}_{k_2} w_2 + \cdots + r^{(m)}_{k_m} w_m, $$
for all $ k_j = 1,2, \ldots, s_j $ and all $ j = 1,2, \ldots, m $. The result then follows by summing these multiplicities.
\end{proof}

\subsection{Comparison with the footprint bound}

In this subsection, we will compare the bounds (\ref{general bound}) and (\ref{eq S-Z bound}) whenever it makes sense to do so. To that end, we will write them as follows: fix a vector of positive weights $ \mathbf{w} = (w_1, w_2, \ldots, w_m) \in \mathbb{N}_+^m $, a positive integer $ r \in \mathbb{N}_+ $, and a polynomial $ F(\mathbf{x}) \in \mathbb{F}[\mathbf{x}] $ such that $ \mathbf{x}^\mathbf{i} = {\rm LM}(F(\mathbf{x})) $, $ \mathbf{i} = (i_1, i_2, \ldots, i_m) $, with respect to the lexicographic ordering. We first consider the footprint bound as in Corollary \ref{corollary footprint for weighted multi}:
\begin{equation}
\# \mathcal{V}_{\geq r, \mathbf{w}}(F(\mathbf{x})) \cdot {\rm B}(\mathbf{w}; r) \leq \# \Delta \left( \left\langle \left\lbrace F(\mathbf{x}) \right\rbrace \bigcup \left\lbrace \prod_{j=1}^m G_j(x_j)^{r_j} : \sum_{j=1}^m r_jw_j \geq r \right\rbrace \right\rangle \right).
\label{eq comparison 1st bound}
\end{equation}
And next we consider the bound (\ref{eq S-Z bound}) as follows:
\begin{equation}
\# \mathcal{V}_{\geq r, \mathbf{w}}(F(\mathbf{x})) \cdot r \leq \#S \sum_{j=1}^m \frac{i_j w_j}{\# S_j}.
\label{eq comparison 2nd bound}
\end{equation}

First we observe that the bound (\ref{eq comparison 1st bound}) also holds for any other monomial ordering, and not only the lexicographic one, as is the case with (\ref{eq comparison 2nd bound}). Second we observe that (\ref{eq comparison 2nd bound}) gives no information whereas (\ref{eq comparison 1st bound}) does, whenever
\begin{equation}
\sum_{j=1}^m \left\lfloor \frac{i_j}{\# S_j} \right\rfloor w_j < r \leq \sum_{j=1}^m \frac{i_j w_j}{\# S_j},
\label{eq region where no information}
\end{equation}
by the discussion in Subsection \ref{subsec interpretation bound}.

Next, we observe that when we do not count multiplicities, that is, $ w_1 = w_2 = \ldots = w_m = r = 1 $, then (\ref{eq comparison 1st bound}) implies (\ref{eq comparison 2nd bound}) via Theorem \ref{th generalization demillo lipton et al}:

\begin{proposition} \label{proposition zippel implies schwartz}
If $ w_1 = w_2 = \ldots = w_m = r = 1 $, that is, $ \mathcal{J} = \{ \mathbf{0} \} $, it holds that $ {\rm B}(\mathbf{w}; r) = 1 $ and 
$$ \# \Delta \left( \left\langle F(\mathbf{x}), G_1(x_1), G_2(x_2), \ldots, G_m(x_m) \right\rangle \right) \leq \# S - \prod_{j=1}^m \left( \# S_j - i_j \right) \leq \#S \sum_{j=1}^m \frac{i_j}{\# S_j}. $$
In particular, (\ref{eq comparison 1st bound}) implies (\ref{eq comparison 2nd bound}) in this case.
\end{proposition}

Moreover, when $ m = 1 $ and we count multiplicities, all bounds coincide, giving (\ref{eq in intro 2}). In the following example we show that this is not the case in general. As we will see, each bound, (\ref{eq comparison 1st bound}) and (\ref{eq comparison 2nd bound}), can be tighter than the other one in different cases, hence complementing each other:

\begin{example}\label{ex1}
Consider $ m = 2 $, $ w_1 = 2 $, $ w_2 = 3 $, $ r = 5 $ and $ \# S_1 = \# S_2 = 4 $. Thus we have that
$$ \mathcal{J} = \{(0,0), (1,0),(0,1), (2,0)\}, \quad \textrm{and} $$
$$ \mathcal{J}_S = \left( [0,11] \times [0,3] \right) \cup \left( [0,3] \times [0,7] \right) . $$
Consider all pairs $ (i_1, i_2) \in \mathcal{J}_S $ and polynomials $ F(x_1,x_2) $ such that $ {\rm LM}(F(x_1,x_2)) = x_1^{i_1} x_2^{i_2} $, with respect to the lexicographic ordering. In Figure \ref{figone}, we show the upper bounds on the number of zeros of $ F(x_1, x_2) $ of weighted multiplicity at least $ 5 $ given by (\ref{eq comparison 1st bound}) (table above) and (\ref{eq comparison 2nd bound}) (table below), respectively. As is clear from the figure, in some regions of the set $ \mathcal{J}_S $, the first bound is tighter than the second (bold numbers in the table above) and vice versa (bold numbers in the table below). Furthermore the first bound gives non-trivial information in the region given by (\ref{eq region where no information}), where the second does not (depicted by dashes).
\end{example}

\begin{figure}[h]
\begin{center}

\begin{tabular}{r|rrrrrrrrrrrr}
$x_1^7$ & \bf{15}&\bf{15}&\bf{15}&\bf{15}\\
$x_1^6$ & 14&\bf{14}&\bf{15}&\bf{15}\\
$x_1^5$ & 13&13&\bf{14}&\bf{15}\\
$x_1^4$ & 12&13&14&15\\
$x_1^3$ & 9&10&11&12&14&\bf{14}&\bf{14}&\bf{14}&\bf{15}&\bf{15}&\bf{15}&\bf{15}\\
$x_1^2$ & 6&7&9&10&12&12&\bf{13}&\bf{13}&\bf{14}&\bf{14}&\bf{15}&\bf{15}\\
$x_1$ & 3&4&6&8&10&10&\bf{11}&\bf{12}&\bf{13}&\bf{13}&\bf{14}&\bf{15}\\
$1$ & 0&2&4&6&8&9&10&11&12&\bf{13}&\bf{14}&\bf{15}\\
\hline
 & $1$ & $x_2$ & $x_2^2$ & $x_2^3$ & $x_2^4$ & $x_2^5$ & $x_2^6$ & $x_2^7$ & $x_2^8$ & $x_2^9$ & $x_2^{10}$ & $x_2^{11}$  \\
  & \\
$x_1^7$ & --&--&--&--\\
$x_1^6$ & 14&--&--&--\\
$x_1^5$ & \bf{12}&13&15&--\\
$x_1^4$ & \bf{9}&\bf{11}&\bf{12}&\bf{14}\\
$x_1^3$ & \bf{7}&\bf{8}&\bf{10}&12&\bf{13}&15&--&--&--&--&--&--\\
$x_1^2$ & \bf{4}&\bf{6}&\bf{8}&\bf{9}&\bf{11}&12&14&--&--&--&--&--\\
$x_1$ & \bf{2}&4&\bf{5}&\bf{7}&\bf{8}&10&12&13&15&--&--&--\\
$1$ & 0&\bf{1}&\bf{3}&\bf{4}&\bf{6}&\bf{8}&\bf{9}&11&12&14&--&--\\
\hline
 & $1$ & $x_2$ & $x_2^2$ & $x_2^3$ & $x_2^4$ & $x_2^5$ & $x_2^6$ & $x_2^7$ & $x_2^8$ & $x_2^9$ & $x_2^{10}$ & $x_2^{11}$
\end{tabular}

\end{center}
\label{figone}
\caption{Upper bounds on the number of zeros of weighted multiplicity at least $ r = 5 $ when $ w_1 = 2 $, $ w_2 = 3 $ and $ \#S_1 = \#S_2 = 4 $, from Example \ref{ex1}. }
\end{figure}

\section*{Appendix: Proof of Lemma \ref{lemma hermite}} \label{appendix}

In this appendix, we give the proof of Lemma \ref{lemma hermite}. We first treat the univariate case ($ m = 1 $) in the classical form. The proof for Hasse derivatives can be directly translated from the result for classical derivatives:

\begin{lemma}
Let $ a_1, a_2, \ldots, a_n \in \mathbb{F} $ be pair-wise distinct and let $ M \in \mathbb{N}_+ $. There exist polynomials $ F_{i,j}(x) \in \mathbb{F}[x] $ such that 
$$ F_{i,j}^{(k)}(a_l) = \delta_{i,k} \delta_{j,l}, $$
for all $ i,k = 0,1,2, \ldots, M $ and all $ j,l = 1,2, \ldots, n $, where $ \delta $ denotes the Kronecker delta.
\end{lemma}

Now, since $ \mathcal{J} $ is finite, we may fix an integer $ M $ such that $ \mathcal{J} \subseteq [0,M]^m $. Similarly, we may find a finite set $ S \subseteq \mathbb{F} $ such that $ T \subseteq S^m $. Denote then $ s = \# S $ and $ S = \{ a_1, a_2, \ldots, a_s \} $, and let $ F_{i,j,k}(x_k) \in \mathbb{F}[x_k] $ be polynomials as in the previous lemma in each variable $ x_k $, for $ i = 0,1,2, \ldots, M $, $ j = 1,2, \ldots, s $ and $ k = 1,2, \ldots, m $. Define now 
$$ F_{\mathbf{i},\mathbf{j}}(\mathbf{x}) = F_{i_1,j_1,1}(x_1) F_{i_2,j_2,2}(x_2) \cdots F_{i_m,j_m,m}(x_m) \in \mathbb{F}[\mathbf{x}], $$
for $ \mathbf{i} = (i_1, i_2, \ldots, i_m) \in [0,M]^m $ and $ \mathbf{j} = (j_1, j_2, \ldots, j_m) \in [1,s]^m $. By the previous lemma and Lemma \ref{lemma leibniz formula}, we see that
$$ F_{\mathbf{i},\mathbf{j}}^{(\mathbf{k})} \left( a_{l_1}, a_{l_2}, \ldots, a_{l_m} \right) = \left( \delta_{i_1, k_1} \delta_{i_2, k_2} \cdots \delta_{i_m, k_m} \right) \left( \delta_{j_1, l_1} \delta_{j_2, l_2} \cdots \delta_{j_m, l_m} \right) = \delta_{\mathbf{i}, \mathbf{k}} \delta_{\mathbf{j}, \mathbf{l}}, $$
for all $ \mathbf{i}, \mathbf{k} \in [0,M]^m $ and all $ \mathbf{j}, \mathbf{l} \in [1,s]^m $. Finally, given values $ b_{\mathbf{i},\mathbf{j}} \in \mathbb{F} $, for $ \mathbf{i} \in \mathcal{J} $ and $ \mathbf{j} \in T $, define 
$$ F(\mathbf{x}) = \sum_{\mathbf{i} \in \mathcal{J}} \sum_{\mathbf{j} \in T} b_{\mathbf{i},\mathbf{j}} F_{\mathbf{i}, \mathbf{j}}(\mathbf{x}) \in \mathbb{F}[\mathbf{x}]. $$
We see that $ {\rm Ev}(F(\mathbf{x})) = ((b_{\mathbf{i},\mathbf{j}})_{\mathbf{i} \in \mathcal{J}})_{\mathbf{j} \in T} $, and we are done.

\section*{Acknowledgement}

The authors gratefully acknowledge the support from The Danish Council for Independent Research (Grant No. DFF-4002-00367). The second listed author also gratefully acknowledges the support from The Danish Council for Independent Research via an EliteForsk-Rejsestipendium (Grant No. DFF-5137-00076B).

\bibliographystyle{plainnat}

\end{document}